\documentclass{article}

\PassOptionsToPackage{numbers, compress}{natbib}
 \usepackage[preprint]{neurips_2026}

\usepackage[utf8]{inputenc} 
\usepackage[T1]{fontenc}    
\usepackage[hidelinks]{hyperref}       
\usepackage{url}            
\usepackage{booktabs}       
\usepackage{amsfonts}       
\usepackage{nicefrac}       
\usepackage{microtype}      
\usepackage{xcolor}         

\usepackage[english]{babel}

\usepackage{tikz}
\usepackage{lipsum}

\usepackage{amsthm, color, float, graphicx, verbatim, placeins, needspace}
\usepackage{comment}

\usepackage{bm}
\usepackage{dsfont}
\usepackage{mathrsfs}
\usepackage{physics}
\usepackage{amsmath, amssymb, mathtools}
\usepackage{thmtools, thm-restate}
\usepackage[roman,full]{complexity}

\usepackage{enumerate}
\usepackage{algpseudocode}
\usetikzlibrary{fit,calc}

\usepackage[linesnumbered,ruled,vlined]{algorithm2e}
\SetKw{Break}{break}

\newtheorem{theorem}{Theorem}

\newtheorem{lemma}[theorem]{Lemma}
\newtheorem{definition}[theorem]{Definition}

\newcommand{\mc}[1]{\mathcal{#1}}
\newcommand{\mb}[1]{\mathbb{#1}}

\newcommand{\supp}{\mathrm{Supp}}

\newcommand{\ope}[1]{\left\|{#1}\right\|_{\mathrm{op}}}

\newcommand{\dtv}{d_\mathrm{TV}}

\newcommand{\pr}{\mathrm{Pr}}

\title{
Winning Lottery Tickets in Neural Networks via a Quantum-Inspired Classical Algorithm
}

\author{%
  Natsuto Isogai \\
  The University of Tokyo \\
  \texttt{natsuto.isogai@g.ecc.u-tokyo.ac.jp} \\
  \And
  Hayata Yamasaki \\
  The University of Tokyo \\
  \texttt{hayata.yamasaki@gmail.com} 
  \And
  Sho Sonoda \\
  RIKEN AIP and
  \\CyberAgent\\
  \And
  Mio Murao \\
  The University of Tokyo \\
  }

\begin{document}

\maketitle

\begin{abstract} 
Quantum machine learning (QML) aims to accelerate machine learning tasks by exploiting quantum computation. Previous work studied a QML algorithm for selecting sparse subnetworks from large shallow neural networks. Instead of directly solving an optimization problem over a large-scale network, this algorithm constructs a sparse subnetwork by sampling hidden nodes from an optimized probability distribution defined using the ridgelet transform. The quantum algorithm performs this sampling in time $O(D)$ in the data dimension $D$, whereas a naive classical implementation relies on handling exponentially many candidate nodes and hence takes $\exp[O(D)]$ time. In this work, we construct and analyze a quantum-inspired fully classical algorithm for the same sampling task. We show that our algorithm runs in time $O(\operatorname{poly}(D))$, thereby removing the exponential dependence on $D$ from the previous classical approach. Numerical simulations show that the proposed sampler achieves empirical risk comparable to exact sampling from the optimized distribution and substantially lower than sampling from the non-optimized uniform distribution, while also exhibiting exponentially improved runtime scaling compared with the conventional classical implementation. These successful dequantization results show that sparse subnetwork selection via optimized sampling can be achieved classically with polynomial data-dimension scaling on conventional computers without quantum hardware, providing an alternative to the existing quantum algorithm.
\end{abstract}

\section{Introduction}
State-of-the-art neural networks achieve strong predictive performance, but their size makes training increasingly costly.
This has led to extensive interest in identifying sparse trainable subnetworks that preserve much of the performance of the original dense networks.
Representative routes include pruning after training~\citep{han2016deepcompressioncompressingdeep}, pruning at initialization~\citep{lee2018snip,10.5555/3495724.3496259}, and dynamic sparse training~\citep{10.5555/3524938.3525214}.
The lottery ticket hypothesis further suggests that large randomly initialized dense networks may contain sparse subnetworks that can be trained in isolation and achieve performance comparable to that of the original dense network~\citep{frankle2018the,liu2024surveylotterytickethypothesis}.

Many existing routes to sparse subnetworks are formulated as optimization procedures over network weights, and thus, sparse-subnetwork construction is often treated as a search problem over weights, masks, or connectivity patterns by handling the original neural network.
By contrast, the random feature method~\citep{NIPS2007_013a006f,NIPS2008_0efe3284,JMLR:v18:15-178}, originally developed for accelerating kernel methods, motivates a different route for constructing such sparse networks.
Ridgelet representations of shallow, single-hidden-layer neural networks provide a natural analytic viewpoint for such a route~\citep{SONODA2017233}.
In the ridgelet viewpoint, the target function is represented as a linear combination of basis functions corresponding to nodes in the hidden layer, and the corresponding coefficients quantify the importance of each hidden node.
This makes sparse-subnetwork construction amenable to importance sampling.
Rather than searching for a sparse subnetwork, one may define a probability distribution that places larger mass on more important hidden nodes and construct a sparse network directly by drawing samples from this distribution.
Each sample corresponds to one hidden node, and thus, the constructed network size is determined by the number of samples.

A realization of this idea was previously proposed in~\citep{10.5555/3618408.3620034} using quantum computation, where the authors introduce the quantum ridgelet transform and use it to construct sparse single-hidden-layer neural networks.
Their approach defines an optimized probability distribution over hidden-node parameters, using given data. Then, it samples from this distribution to construct a sparse network.
This provides a sampling-based route to sparse-subnetwork construction, motivated by the lottery ticket hypothesis but distinct from standard pruning-based routes.

The existing approach to this sampling procedure relies on quantum computation.
Quantum machine learning (QML) has been studied as a field aiming to accelerate machine learning by using quantum computation, which has the capability of implementing high-dimensional linear transformations via the technique of quantum singular-value transformation (QSVT)~\citep{wittek2014, Biamonte2017, PhysRevLett.103.150502, gilyen2019quantum}.
In the ridgelet-based subnetwork construction of~\citep{10.5555/3618408.3620034}, quantum computation is used for two operations: implementing the ridgelet transform and applying QSVT to implement the inverse of a matrix with regularization, which is required to define the optimized sampling distribution.
However, an implementation of such a QML algorithm requires quantum hardware capable of preparing and manipulating coherent quantum states.
Large-scale fault-tolerant quantum computation (FTQC) capable of implementing such primitives is not yet achieved on current quantum hardware, and realistic implementations may incur substantial overhead from noise suppression and quantum error correction~\citep{Preskill2018quantumcomputingin, Acharya2025}; therefore, a classical, i.e., non-quantum, algorithm running on conventional computational hardware for the same probability distribution may be valuable when its overhead is moderate, e.g., polynomial overhead.

This viewpoint has also motivated extensive work on quantum-inspired classical algorithms, which asks whether algorithmic ideas originally developed for quantum computation can be converted into efficient classical algorithms with only moderate overhead.
In general, classical simulation of quantum computation is not always expected to be efficient, and this belief is one of the motivations for quantum computation~\citep{N4, arora2009computational}.
Nevertheless, a series of quantum-inspired classical algorithms has shown that, under suitable input models or structural assumptions, certain machine-learning tasks initially solved by polynomial-time quantum algorithms can also be solved by polynomial-time classical algorithms~\citep{10.1145/3313276.3316310, chia2022sampling, le2025robust}.
In this sense, we say that a quantum algorithm is ``dequantized'' when its computational task is solved by a classical algorithm with at most polynomial overhead.

These results suggest that when a quantum algorithm relies on the data-access assumption or task-specific structure, the corresponding sampling procedure or linear-algebraic transformation may sometimes admit an efficient classical algorithm.
This motivates the central question of this work: whether quantum computation is essential for the above sampling, i.e., whether a classical algorithm can achieve the sampling task with a performance comparable to the quantum algorithm in~\citep{10.5555/3618408.3620034}.

\subsection{Problem}
The ridgelet-based route turns sparse-subnetwork construction into a sampling task.
The optimized distribution of~\citep{10.5555/3618408.3620034} assigns probability mass to hidden-node parameters according to the magnitudes of their regularized ridgelet coefficients.
Sampling from this distribution is a natural way to select important hidden nodes, where each sampled parameter becomes one node in the sparse network.

The primary difficulty in this sampling task is that the regularized ridgelet coefficients are not directly available.
They are expressed through an inverse of a matrix with regularization
\begin{equation}\label{eq: ridgelet coefficients}
    (\bm{R}\hat{\bm{P}}_{\mathrm{data}}\bm{R}^\top + \lambda P^{-D} \bm{I})^{-1} \bm{R} \ket{\psi_\mathrm{in}},
\end{equation}
where $\bm{R}$ denotes the discrete ridgelet transform, and $\hat{\bm{P}}_{\mathrm{data}}$ is a matrix determined by the empirical data distribution.
This inverse acts on the space of hidden-node parameters, whose size grows exponentially with the dimension $D$ of the input data.
Moreover, the matrix to be inverted is not sparse or low rank in general.
Thus, a direct classical implementation involves computing the inverse of a dense $\exp[O(D)]\times\exp[O(D)]$ matrix, which usually takes exponential time in $D$.

The second obstacle is that the ability to compute the unnormalized weight at each hidden node is not enough to obtain a sampler.
Even if one can compute the score of a candidate hidden node, sampling from the optimized distribution would still require normalizing these weights over the entire space.
The normalization constant is a sum over all hidden-node parameters and thus cannot be efficiently computed by direct enumeration over the exponentially large parameter space.
Furthermore, a naive rejection sampler based on a uniform proposal distribution cannot be efficient since the target distribution may be highly concentrated on a small subset of important nodes.

The quantum algorithm in~\citep{10.5555/3618408.3620034} overcomes these obstacles by implementing the relevant high-dimensional transformations as operations on quantum states rather than manipulating classical vectors.
Using the quantum ridgelet transform and QSVT, it prepares a quantum state whose measurement produces samples from a probability distribution close to the optimized distribution.

A natural approach to dequantization would, therefore, be to apply existing quantum-inspired techniques to this QSVT-based procedure.
However, existing quantum-inspired classical algorithms do not apply directly to the matrix in~\eqref{eq: ridgelet coefficients}.
Such techniques require specific assumptions, such as low-rank structure or sparsity of the relevant matrices~\citep{10.1145/3313276.3316310,chia2022sampling,le2025robust}.
In contrast, the matrix $\bm{R}\hat{\bm{P}}_{\mathrm{data}}\bm{R}^\top + \lambda P^{-D} \bm{I}$ is full rank and dense since the ridgelet transform $\bm{R}$ is dense, and also, we do not assume that the resulting matrix is low rank in the parameter space.
Therefore, conventional quantum-inspired techniques based on standard low-rank or sparse matrices are not directly applicable.

\subsection{Main Contributions}
The results of this paper are summarized in Table~\ref{tab:comparison}.
\paragraph{A classical algorithm for sampling from the optimized distribution.}
Our main result is a randomized classical algorithm for sampling from the optimized probability distribution used in the ridgelet-based sparse-subnetwork construction in~\citep{10.5555/3618408.3620034}.
Given the dataset and the regularization parameter $\lambda$, our algorithm outputs a hidden-node parameter by sampling from the optimized distribution determined by the regularized ridgelet coefficients while avoiding the above obstacles by explicitly exploiting the structure of the discrete ridgelet transform.
As summarized in Table~\ref{tab:comparison}, this gives a polynomial-time classical sampler in the regime where the number of training samples $M = \poly(D)$ and the other relevant parameters are polynomially bounded.

\paragraph{An improved quantum algorithm for sampling from the optimized distribution.}
The technique used to construct our classical algorithm also improves the quantum algorithm for this sampling.
In the quantum algorithm in~\citep{10.5555/3618408.3620034}, the matrix inversion with regularization in~\eqref{eq: ridgelet coefficients} is implemented with a runtime depending on the regularization parameter $1/\lambda$.
Our analysis shows that this generic inverse step is not necessary for this sampling task.
Therefore, we eliminate the $1/\lambda$ dependence in the runtime of our classical algorithm and the corresponding improved quantum algorithm.

\vspace{-0.2cm}
\begin{table}[h]
\centering
\small
\caption{
Comparison of approaches for sampling from the optimized ridgelet distribution.
}
\vspace{-0.2cm}
\label{tab:comparison}
\begin{tabular}{@{}lcccc@{}}
\toprule
Approach 
& End-to-end polytime 
& Preprocessing 
& $M$ dependence
& $1/\lambda$ dependence \\
\midrule
Naive classical 
& No 
& None 
& -- 
& Yes \\

Previous quantum~\citep{10.5555/3618408.3620034} 
& Yes 
& $\widetilde O(M)$ 
& $\polylog(M)$ 
& Yes \\

This work, classical 
& Yes 
& $\widetilde O(M)$ 
& $\poly(M)$ 
& No \\

This work, improved quantum 
& Yes 
& $\widetilde O(M)$ 
& $\polylog(M)$ 
& No \\
\bottomrule
\end{tabular}
\end{table}
\vspace{-0.2cm}

\paragraph{Numerical experiments.}
We numerically demonstrate the performance of our classical algorithm; after sampling hidden nodes, we train the output weights on these nodes.
The resulting sparse networks achieve empirical risk comparable to exact sampling from the optimized distribution and substantially better than sampling from the uniform distribution (Fig.~\ref{fig: Comparison of empirical main}).
We also numerically verify that our classical algorithm achieves polynomial runtime scaling, while the naive implementation of exact sampling from the same optimized distribution requires exponential time (Fig.~\ref{fig: Comparison of runtimes}).

\paragraph{Implications.}
Our results open an efficient classical route to finding a lottery-ticket sparse network on conventional hardware via optimized sampling of trainable hidden nodes, without directly handling the large-scale dense network or relying on quantum hardware.
More broadly, the techniques that we develop dequantize a QML subroutine by exploiting problem-specific structure beyond the assumptions used in existing dequantization techniques, providing complementary tools for assessing when apparent quantum advantages in QML can be reproduced by efficient classical algorithms.

\section{Settings}\label{sec: settings}
In this section, we define the learning setting used in this work.
First, we define a ridgelet representation for a single-hidden-layer neural network over the discretized domain.
Then, we introduce the empirical risk minimization and its optimal solution based on the representation.
Finally, we define a probability distribution optimized for constructing a sparse subnetwork and clarify the sampling problem we address in this work.
Appendix~\ref{sec:preliminaries} summarizes basic notions and notations of quantum computation to describe our results, referring to the standard textbook~\citep{N4} for more details.
Following the bra-ket notation, we may represent a vector as a ket.
\subsection{Discrete Ridgelet Transform}
In this work, we consider a supervised regression problem.
Let $D$ be an input dimension, and $f: \mathbb{R}^D \to \mathbb{R}$ be an unknown target function to be learned.
Let $g: \mathbb{R} \to \mathbb{R}$ denote an activation function such as rectified linear unit (ReLU) and sigmoid.
A single-hidden-layer fully connected neural network approximates the target function $f$ in the form
$f(\bm{x}) \approx \sum_{n = 1}^{N} w_n g(\bm{a}_n^\top \bm{x} - b_n)$,
where $N$ is the number of hidden-layer nodes, and $w_n$ is a weight associated with node for $(\bm{a}_n,b_n)$.

Our goal is not to optimize all hidden nodes directly, but to construct a sparse subnetwork with a small number of hidden nodes that strongly affect the approximation.
In the over-parameterized limit $N \to \infty$, a one-hidden-layer neural network is simplified into an integral representation
$\mathcal{S}[w](\bm{x}) = \int_{\mathbb{R}^D \times \mathbb{R}} d(\bm{a}) db w(\bm{a},b)g(\bm{a}^\top \bm{x} - b)$,
and by using a function $r:\mathbb{R} \to \mathbb{R}$, we define a ridgelet transform
$\mathcal{R}[f](\bm{a},b) = \int_{\mathbb{R}^D} d(\bm{x}) f(\bm{x})r(\bm{a}^\top \bm{x} - b)$~\citep{SONODA2017233,256500,MURATA1996947,candes1998ridgelets}.
Under a suitable admissibility condition on $g$ and $r$, the reconstruction formula
$\mathcal{S}[\mathcal{R}[f]] = f$ holds~\citep{SONODA2017233},
and thus, we can consider $\mathcal{R}[f](\bm{a},b)$ as a coefficient contributing to the evaluation of the neural network for each node parameterized by $(\bm{a},b)$.

In this work, we perform analysis over discretized models rather than continuous models.
This is not merely an assumption for convenience.
In actual computation, inputs and parameters are represented with finite precision, e.g., one uses a bitstring to represent a real number.
In fact, the formulation of the quantum ridgelet transform in previous research is also based on a discrete setting using a finite field~\citep{10.5555/3618408.3620034}.
Therefore, we use a discrete analog of this representation over $\mathbb{Z}_P \coloneq \{0,1,\ldots, P-1\}$, where $P$ is a prime number representing a precision parameter to approximate the real-valued space $\mathbb{R}$ by $\mathbb{Z}_P$.
In the following, when we write a sum, $\bm{x}$ and $\bm{a}$ run over $\mathbb{Z}_P^D$, and $y$ and $b$ over $\mathbb{Z}_P$ unless specified otherwise.

\textbf{Discrete ridgelet representation:}
We suppose both an activation and a ridgelet function satisfy the following conditions over the domain $\mathbb{Z}_P$ instead of $\mathbb{R}$: (1) an activation function $g:\mathbb{Z}_P \to \mathbb{R}$ is assumed to be normalized as $\sum_{b} g(b) = 0$, $\|g\|_2^2 = \sum_{b } |g(b)|^2 = 1$, (2) we can choose $r$ corresponding to $g$, satisfying a certain admissibility condition~\citep{SONODA2017233,10.5555/3618408.3620034}, but, in this work, we choose a ridgelet function $r:\mathbb{Z}_P \to \mathbb{R}$ as $r = g$ for simplicity of notation.

In the above condition, by representing the integral over the continuous space $\mathbb{R}$ by the summation over the finite field $\mathbb{Z}_P$, the discretized neural network for $w: \mathbb{Z}_P^D \times \mathbb{Z}_P \to \mathbb{R}$ and the discrete ridgelet transform of $f: \mathbb{Z}_P^D \to \mathbb{R}$ are defined as, respectively,
$\mathcal{S}[w](\bm{x}) \coloneq P^{-\frac{D}{2}}\sum_{\bm{a},b} w(\bm{a},b) g((\bm{a}^\top \bm{x} - b) \bmod P)$,
    and
$\mathcal{R}[f](\bm{a},b) \coloneq P^{-\frac{D}{2}}\sum_{\bm{x}} f(\bm{x})r((\bm{a}^\top \bm{x} - b) \bmod P)$,
where the coefficient $P^{-\frac{D}{2}}$ works as a normalization constant.
The discretization of the neural network and the ridgelet transform preserves the reconstruction formula $f(\bm{x}) = \mathcal{S}[\mathcal{R}[f]](\bm{x})$.
In particular, the ridgelet coefficients $\mathcal{R}[f]$ provide an exact representation of the target function in the discretized single-hidden-layer network~\citep{10.5555/3618408.3620034}.
It is convenient to represent the discrete ridgelet transform as a matrix $\bm{R} \coloneq P^{-D/2} \sum_{\bm{x},\bm{a},b}
    r((\bm{a}^\top \bm{x} - b)\bmod P)\ket{\bm{a},b}\bra{\bm{x}}$.
Then, for any function $f:\mathbb{Z}_P^D\to\mathbb{R}$, if we write $\ket{f} \coloneq \sum_{\bm{x}} f(\bm{x})\ket{\bm{x}}$, we have $\bm{R}\ket{f} = \sum_{\bm{a},b}
    \mathcal{R}[f](\bm{a},b)\ket{\bm{a},b}$.

\subsection{Regularized Ridgelet Estimator and Optimized Distribution}
\paragraph{Data and empirical objective.}
We consider a supervised regression task on the discrete input domain over $\mathbb{Z}_P^D$.
This discrete setting can be viewed as a finite-precision model of continuous-valued data.
Let $f: \mathbb{Z}_P^D \to \mathbb{R}$ be an unknown target function to be learned, and let $p_\mathrm{data}: \mathbb{Z}_P^D \to \mathbb{R}_{\geq 0}$ be the input distribution.
We are given a dataset of $M$ input-output pairs $\{(\bm{x}_m,y_m)\}_{m \in [M]} \subseteq \mathbb{Z}_P^D \times \mathbb{R}$, where each input $\bm{x}_m$ is drawn from $p_\mathrm{data}$ and each label is given by $y_m = f(\bm{x}_m)$.
The empirical distribution induced by the dataset is $\hat p_{\mathrm{data}} \coloneq (1/M) \sum_{m \in [M]} \bm{1}\{\bm{x}_m = \bm{x}\}$.
We aim to approximate $f$ by a single-hidden-layer discretized neural network $\hat{f} = \mathcal{S}[w]$ from the training dataset.
We measure the fit of $\hat{f}$ by the weighted empirical squared loss $J(w) \coloneq \sum_{\bm{x} \in \mathbb{Z}_P^D} \hat{p}_\mathrm{data}(\bm{x}) | f(\bm{x}) - \hat{f}(\bm{x}) |^2$.
To obtain a stable coefficient vector, we add the $\ell_2$ penalty $\Omega(w) = \|P^{-\frac{D}{2}}w\|_2^2 = \sum_{\bm{a},b} |P^{-\frac{D}{2}}w(\bm{a},b)|^2$ with a hyperparameter $\lambda$, and we consider a regularized least-squares regression problem, i.e., finding the optimal estimator $w_\lambda^\ast$ satisfying
\begin{equation}\label{eq : optimal estimator}
    w_\lambda^\ast \coloneq \operatorname*{arg\,min}_{w} \{\tilde{J}(w)\},
\end{equation}
where $\tilde{J}(w) \coloneq J(w) + \lambda \Omega(w)$.
The regularization makes the objective function strongly convex in the coefficient vector and fixes a unique minimizer.
In particular, choosing $\lambda = O(\poly(\epsilon))$ makes the bias term $O(\epsilon)$, which guarantees the final empirical approximation guarantee $J(w) = O(\epsilon)$. 
We also define 
\begin{equation}\label{eq: definition gamma}
    \gamma \coloneq \|P^{-\frac{D}{2}}w_\lambda^\ast\|_2^2
\end{equation}
that measures the scale of the normalized regularized ridgelet coefficients and appears in the runtime bounds of our sampling algorithms.

\paragraph{Optimized Probability Distribution for Sparse-Subnetwork Construction.}
The optimal estimator $w_\lambda^\ast$ in~\eqref{eq : optimal estimator} determines the contribution of each node parameterized by $(\bm{a},b)$ in the discretized ridgelet representation.
To construct a sparse subnetwork, we would like to sample hidden-node parameters with large normalized ridgelet weights.
Following~\citep{10.5555/3618408.3620034}, we define the optimized probability distribution as follows.
\begin{definition}[Optimized probability distribution]\label{def: optimized probabilit distributon}
    For parameters $\lambda > 0$ and $\Delta > 0$, the optimized probability distribution is defined by
    \begin{equation}
        p_{\lambda,\Delta}^\ast(\bm{a},b) \coloneq \frac{1}{Z}\frac{\left|P^{-\frac{D}{2}}w_\lambda^\ast(\bm{a}, b)\right|^2}{\left|P^{-\frac{D}{2}}w_\lambda^\ast(\bm{a}, b)\right|^2 + \Delta}
    \end{equation}
    where $Z$ is a normalization constant ensuring that $p_{\lambda,\Delta}^\ast$ is a probability distribution, i.e., $\sum_{(\bm{a},b) \in \mathbb{Z}_P^D \times \mathbb{Z}_P} p_{\lambda,\Delta}^\ast(\bm{a},b) = 1$.
\end{definition}
This probability distribution assigns large probability masses to high-weight parameters.
The smoothing parameter $\Delta$ smooths the distribution by saturating the high-weight nodes and more strongly suppressing the low-weight nodes: we can choose the probability distribution proportional to the square of optimal weights $w_\lambda^\ast$, but this choice may be highly biased.

Equivalently, using the matrix expression of $w_\lambda^\ast$, we write the quantum state
\begin{equation}\label{eq: optimized quantum state}
    \ket{p_{\lambda,\Delta}^\ast} \propto \left(\bm{W}_\lambda + \frac{\Delta}{\gamma}\bm{I}\right)^{-1/2} \left(\bm{R} \hat{\bm{P}}_\mathrm{data}\bm{R}^\top + \lambda P^{-D} \bm{I} \right)^{-1} \bm{R} \ket{\psi_\mathrm{in}},
\end{equation}
where $\bm{W}_\lambda \coloneq \frac{1}{\gamma} \sum_{\bm{a},b} |P^{-D/2}w_\lambda^\ast(\bm{a},b)|^2 \ketbra{\bm{a},b}$, $\hat{\bm{P}}_\mathrm{data} \coloneq \sum_{\bm{x}} \hat{p}_\mathrm{data}(\bm{x}) \ketbra{\bm{x}}$, and $\ket{\psi_{\mathrm{in}}} \coloneq \hat{\bm{P}}_\mathrm{data} \ket{f}$.
This form is used in the quantum ridgelet transform algorithm~\citep{10.5555/3618408.3620034}, i.e., the quantum algorithm prepares a quantum state close to~\eqref{eq: optimized quantum state} and measures it to sample from the optimized probability distribution.

The computational task studied in this paper is to sample from a distribution close to $p_{\lambda,\Delta}^\ast$ in total variation distance.
More precisely, given $\lambda > 0$, $\Delta > 0$, and an accuracy parameter $\delta \in (0,1)$, our goal is to output one sample from a distribution $\mathcal{D}'$ over $\mathbb{Z}_P^D \times \mathbb{Z}_P$ such that $\dtv(p_{\lambda,\Delta}^\ast, \mathcal{D}')\leq \delta$ in total variation distance.

Once such a sampling algorithm is obtained, repeated sampling yields hidden-node parameters for a sparse single-hidden-layer network.
The approximation guarantee follows from the bounds for finding the lottery-ticket theorem in~\citep{10.5555/3618408.3620034}.
 
\section{Sampling from Optimized Probability Distribution}
We now provide our main results.
Our primary result is a randomized classical algorithm for sampling from the optimized probability distribution $p_{\lambda,\Delta}^\ast$.
We explain the two main ideas behind the sampling algorithm.
Finally, we state the corresponding improvement of the quantum sampling algorithm.

\subsection{Classical sampling algorithm for optimized probability distribution}
Our main result for the classical sampling algorithm is as follows.
\begin{theorem}[Classical sampler for the optimized distribution]\label{thm: ridgelet transform}
    Given the $D$-dimensional dataset $\{\bm{x}_m, y_m\}_{m \in [M]}$, and parameters $\lambda, \Delta > 0$ with $\gamma >0$ in~\eqref{eq: definition gamma}, for any accuracy parameter $\delta \in (0,1)$, there exists a randomized classical algorithm that outputs one sample from a distribution $\mathcal{D}'$ whose total variation distance from $p_{\lambda,\Delta}^\ast$ in Definition~\ref{def: optimized probabilit distributon} is at most $\delta$, with sample runtime
    \begin{equation}
        O(D\log(1/\delta)) \times \tilde{O}({M^2 \gamma}/{\Delta}).
    \end{equation}
\end{theorem}
In particular, if $M = \poly(D)$ and the relevant parameters $\gamma, \Delta^{-1}$ and $\log(1/\delta)$ are polynomially bounded in $D$, then the optimized-distribution sampling task admits a polynomial-time randomized classical implementation.
See Appendix~\ref{sec: proof of ridgelet} for the proof.
We explain the main ideas below.

\subsection{Main Idea of Classical Sampling Algorithm}
The difficulties in this sampling task lie in the QSVT for the high-dimensional matrix in~\eqref{eq: optimized quantum state}.
A natural approach would be to classically simulate the QSVT part by using quantum-inspired techniques.
Our target matrix to be inverted, however, does not satisfy the assumptions required for such approaches, i.e., sparsity or low-rankness.
Our algorithm instead uses a problem-specific structure together with a quantum-inspired sampling-and-query data structure.

\paragraph{Avoidance of the dense inverse.}
The central observation of our algorithm is that the apparent high-dimensional inverse is not a general inverse.
Although $(\bm{R}\hat{\bm{P}}_\mathrm{data}\bm{R}^\top + \lambda P^{-D} \bm{I})^{-1}$ acts on the high-dimensional hidden-node spaces, it has a special ridgelet structure, i.e., $\bm{R}$ is an isometry and $\hat{\bm{P}}_{\mathrm{data}}$ is diagonal.
This allows us to decompose the inverse as
\begin{equation}
    \left(\bm{R}\hat{\bm{P}}_\mathrm{data}\bm{R}^\top + \lambda P^{-D} \bm{I}\right)^{-1} = \bm{R}\left(\hat{\bm{P}}_\mathrm{data} + \lambda P^{-D} \bm{I}\right)^{-1} \bm{R}^\top + \frac{1}{\lambda P^{-D}} \left(\bm{I} - \bm{R}\bm{R}^\top \right).
\end{equation}
The first term acts on the image of $\bm{R}$, while the second term acts on its orthogonal complement.
Since the vector to which the inverse is applied is $\bm{R}\ket{\psi_\mathrm{in}}$, the orthogonal-complement term vanishes $\left(\bm{I}-\bm{R}\bm{R}^\top\right)\bm{R}\ket{\psi_\mathrm{in}}=0$.
Consequently, the regularized ridgelet coefficient vector can be written as
\begin{equation}\label{eq: the representation of w in diagonal inverse}
    w_\lambda^\ast = \bm{R}\ket{\phi},
\end{equation}
where $\ket{\phi} = (\hat{\bm{P}}_\mathrm{data} + \lambda P^{-D} \bm{I})^{-1}\ket{\psi_\mathrm{in}}$.
Thus, the dense inverse over the hidden-node parameter space is reduced to an entrywise diagonal reweighting over the data domain, followed by the ridgelet transform.
This reduction is the key reason why our algorithm avoids implementing the regularized inverse as a general high-dimensional matrix transformation.

After the action of $\bm{R}$, each amplitude becomes a global sum involving all coordinates of the transformed vector weighted by the ridgelet function $r$.
In this sense, the computational bottleneck is a highly complex correlation sum induced by the ridgelet transform rather than the local access to individual entries in $\ket{\phi}$.

This characteristic is reminiscent of the Forrelation problem, where one is asked to estimate the correlation between one Boolean function and the Fourier transform of another~\citep{10.1145/2746539.2746547}.
That problem exhibits an optimal quantum-classical query separation: it can be solved with a constant number of quantum queries, while any randomized classical algorithm requires an exponentially large number of queries in the oracle setting.
Nevertheless, in this work, we avoid this problem by using the sparsity $M$ on the vector $\ket{\phi}$, i.e., the correlation in the amplitude is reduced to a summation of $M$ terms
\begin{equation}
   \sum_{\bm{x}} r(\bm{a}^\top \bm{x} - b \bmod P) \frac{\hat{p}_\mathrm{data}(\bm{x}) f(\bm{x})}{\hat{p}_\mathrm{data}(\bm{x}) + \lambda P^{-D}} = \sum_{\bm{x} \in \mathcal{X}_\mathrm{data}} r(\bm{a}^\top \bm{x} - b \bmod P) \frac{\hat{p}_\mathrm{data}(\bm{x}) f(\bm{x})}{\hat{p}_\mathrm{data}(\bm{x}) + \lambda P^{-D}},
\end{equation}
where $\mathcal{X}_\mathrm{data} \coloneq \supp(\hat{p}_\mathrm{data})$ is a support of the empirical probability distribution.

\paragraph{From weight evaluation to sampling.}
This reduction gives an efficient way to evaluate the unnormalized weight of the optimized probability distribution $p_{\lambda,\Delta}^\ast$ for each hidden-node parameter $(\bm{a},b)$.
However, pointwise evaluation is not enough to sample from $p_{\lambda,\Delta}^\ast$.
In fact, for sampling, the unnormalized weights need to be converted into probabilities.
Thus, a direct sampler requires either computing $Z$ in Definition~\ref{def: optimized probabilit distributon} or avoiding this normalization step.
Moreover, a naive rejection sampling based on the uniform proposal distribution can be inefficient, since the optimized distribution may concentrate on a small subset of high-weight hidden nodes.

To address these problems, we design a proposal distribution suitable for the optimized coefficients.
Let $K \coloneq |\mathcal{X}_\mathrm{data}| \leq M$.
We define 
\begin{equation}\label{eq: proposal distribution}
    q(\bm{a},b) \coloneq \frac{1}{\gamma} \sum_{\bm{x} \in \mathcal{X}_{\mathrm{data}}} \left|P^{-D/2}\frac{\hat{p}_{\mathrm{data}}(\bm{x})f(\bm{x})}{\lambda P^{-D} + \hat{p}_{\mathrm{data}}(\bm{x})}\right|^2\left|P^{-D/2}r\left(\bm{a}^\top \bm{x} - b )\mod P \right) \right|^2,
\end{equation}
which is chosen so that every high-weight hidden node with a large optimized weight is proposed with sufficiently large probability.

This proposal distribution $q$ is efficiently samplable.
To sample from $q$, we use a quantum-inspired data structure for the vector $\ket{\phi}$.
This is the same type of binary-tree data structure used for efficient quantum state preparation and quantum-inspired classical algorithms~\citep{10.1145/3313276.3316310,chia2022sampling,grover2002creatingsuperpositionscorrespondefficiently, kerenidis_et_al:LIPIcs.ITCS.2017.49}.
In fact, the first factor in each summand is proportional to the squared magnitude of $\ket{\phi}$, while the second factor is the probability of a ridgelet transform compatible with the sampled data point.
Thus, to sample from $q$, we first sample a data point $\bm{x}$ according to the squared entries of $\ket{\phi}$.
Then, we sample $\bm{a}$ uniformly from $\mathbb{Z}_P^D$ and sample $t \in \mathbb{Z}_P$ with probability $|r(t)|^2$.
Finally, we set $b = \bm{a}^\top \bm{x} - t \bmod P$.
This procedure produces $(\bm{a},b)$ with probability $q(\bm{a},b)$.
Therefore, the proposal distribution $q$ is not only analytically useful for rejection sampling, but also efficiently samplable.

Given a proposal sample $(\bm{a},b) \sim q$, we accept it with probability
\begin{equation}\label{eq: thinning rate ridgelet}
    \frac{\Delta}{
        \Delta+\left|P^{-D/2}w_\lambda^\ast(\bm a,b)\right|^2
    }
    \frac{
        \left|P^{-D/2}w_\lambda^\ast(\bm a,b)\right|^2
    }{
        K\gamma q(\bm a,b).
    },
\end{equation}
The Cauchy-Schwarz inequality guarantees that this acceptance probability is at most one.
Conditioned on acceptance, the output distribution is exactly $p_{\lambda,\Delta}^\ast$.
Repeating this rejection-sampling procedure sufficiently many times makes the failure probability at most $\delta$.
This yields the randomized classical sampler in Theorem~\ref{thm: ridgelet transform}.

\subsection{Improved quantum sampling algorithm}
The same inverse reduction also improves the quantum sampling algorithm of~\citep{10.5555/3618408.3620034}.
The original quantum algorithm prepares the sampling state in~\eqref{eq: optimized quantum state} by implementing the regularized inverse $(\bm{R}\hat{\bm{P}}_{\mathrm{data}}\bm{R}^\top +\lambda P^{-D}\bm{I})^{-1}$ as a high-dimensional matrix transformation.
This implementation leads to an explicit $1/\lambda$ dependence in the runtime.

Our inverse reduction shows that this inverse step is unnecessary.
In fact, by~\eqref{eq: the representation of w in diagonal inverse}, the coefficient quantum state can be prepared by preparing $\ket{\phi}$ and then applying the quantum ridgelet transform $\bm{R}$.
Importantly, the $\ket{\phi}$ used here is exactly the same data-dependent vector for which our classical algorithm constructs a data structure.
The same structure also provides the quantum state preparation access required by the quantum algorithm.
Therefore, the quantum algorithm no longer needs to implement the regularized inverse over the hidden-node space as a QSVT transformation.
The remaining QSVT step is only the smoothing transformation associated with $\bm{W}_\lambda+\frac{\Delta}{\gamma}\bm{I}$ as in the original sampling procedure.
This yields the following improved quantum sampling algorithm.
The proof is given in Appendix~\ref{sec: proof of quantum sampler}.
\begin{theorem}[Improved quantum sampler for the optimized distribution]\label{thm: improved quantum sampler in main}
    Given the $D$-dimensional dataset $\{\bm{x}_m,y_m\}_{m\in[M]}$ and parameters $\lambda,\Delta>0$, with $\gamma$ defined in~\eqref{eq: definition gamma}, for any sampling precision $\delta\in(0,1)$, there exists a quantum algorithm that outputs one sample from a distribution $\mathcal{D}'$ whose total variation distance from $p_{\lambda,\Delta}^\ast$ in Definition~\ref{def: optimized probabilit distributon} is at most $\delta$, with sample runtime
    \begin{equation}
        \tilde O\!\left((
            {D\gamma}/{\Delta})
            \polylog({MP}/{\delta})
        \right)
    \end{equation}
    after preprocessing time $\tilde O(M)$.
\end{theorem}

\section{Numerical Experiments}\label{sec:numerical-experiments}
In this section, we numerically demonstrate that our algorithm efficiently yields a desired sparse subnetwork.
The proposed sampling algorithms construct the sparse subnetwork by sampling hidden nodes from the optimized probability distribution; i.e., the sampled nodes define the basis functions for the subnetwork.
After sampling these important nodes, we train the output weight on the nodes.
The previous algorithm in~\citep{10.5555/3618408.3620034} for this sampling relied on quantum computation; in contrast, our dequantized sampler in~Theorem~\ref{thm: ridgelet transform} is fully classical and can be executed on a conventional computer.

We numerically evaluate whether the sampled nodes generated by our sampler yield a lower empirical risk than those drawn from the non-optimized uniform distribution.
Moreover, we evaluate its runtime scaling to verify that the proposed algorithm avoids the exponential runtime in the data dimension required by a naive classical implementation of sampling from the optimized distribution.

\paragraph{Empirical-risk comparison.}

\begin{figure}[t]
  \centering
  \includegraphics[width=\linewidth]{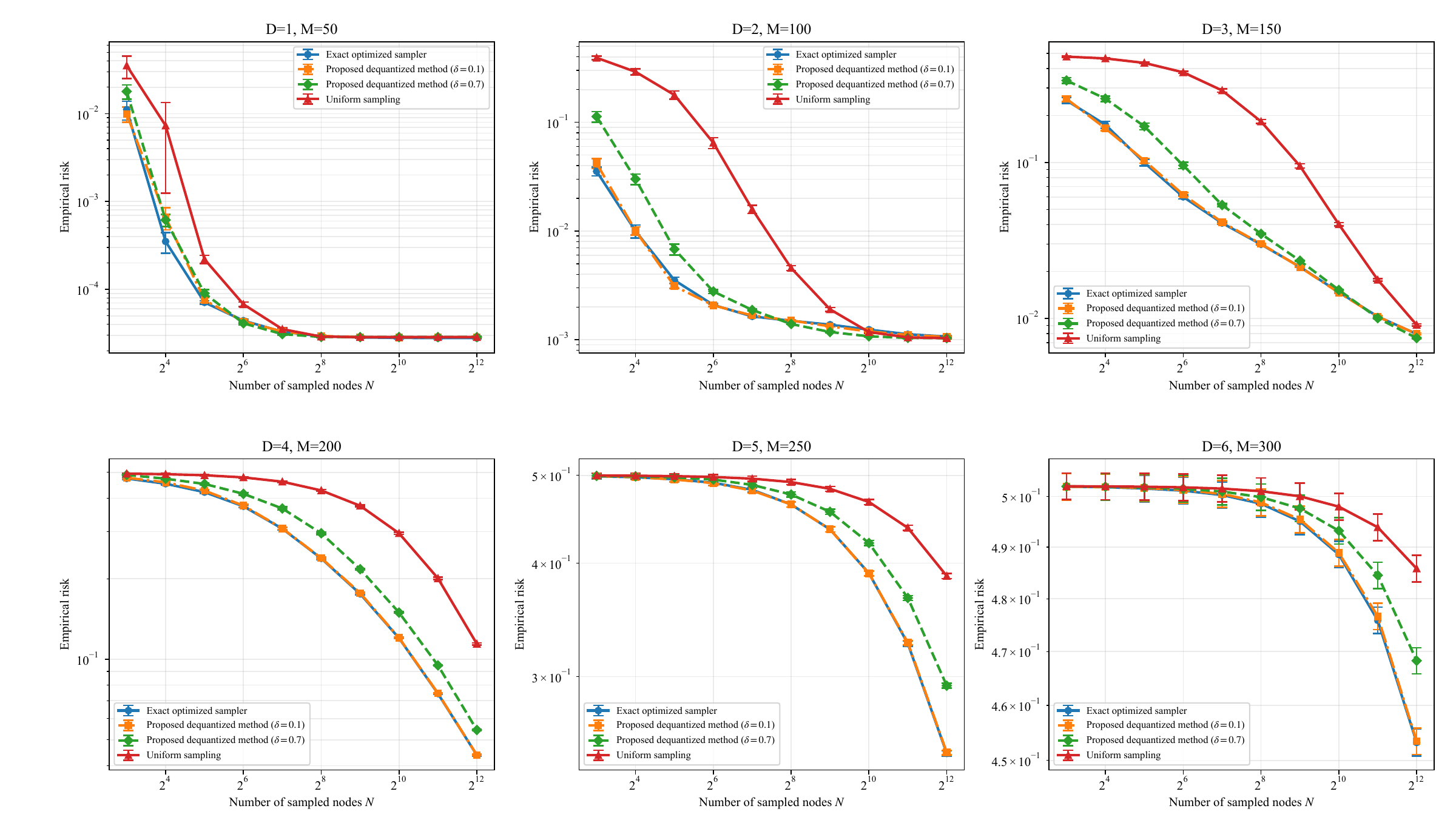}
  \caption{Empirical risk of sparse subnetworks constructed from $N$ sampled hidden nodes. We compare exact sampling from the optimized distribution (blue solid line), our dequantized classical algorithm with $\delta = 0.1$ (orange dashdot line), $\delta = 0.7$ (green dashed line), and sampling from the non-optimized uniform distribution (red solid line). The experiments use $N \in \{8, 16, \ldots, 4096\}$. Curves show the mean over $50$ executions, while error bars are the standard error of the mean.}
  \label{fig: Comparison of empirical main}
\end{figure} 

For the empirical-risk experiment, we work on the finite domain $\mathbb{Z}_P^D$ with $P=7$ and
$D\in\{1,\ldots,6\}$.
We set the function to be learned as a sine function.
For each $D$, we draw $M=50D$ training inputs independently from the uniform distribution $p_\mathrm{data}(\bm{x}) = 1/P^D$, and the empirical distribution $\hat{p}_{\mathrm{data}}$ is induced by the training dataset.
We use ReLU as the activation function, and the ridgelet function $g = r$.
We set $\lambda P^{-D}=10^{-3}$, and set $\Delta=\gamma$ for efficiency.
Then, we compare exact sampling from the optimized distribution $p_{\lambda,\Delta}^{\ast}$ in Definition~\ref{def: optimized probabilit distributon}, our sampler approximately sampling from $p_{\lambda,\Delta}^{\ast}$ with accuracy parameters $\delta=0.1$ and $\delta=0.7$, and sampling from the uniform distribution.
Given $N$ sampled hidden nodes $\{(\bm{a}_n,b_n)\}_{n=1}^N$, we fit the weights by ridge regression and evaluate the empirical risk $\sum_{\bm{x}} \hat p_{\mathrm{data}}(\bm{x}) | f_D(\bm{x})-\hat{f_N}(\bm{x}) |^2$, where $\hat{f}_N$ is a trained subnetwork constructed from $N$ sampled nodes.
This ridge estimator is obtained by solving the corresponding normal equations using linear algebra routines in NumPy.
We plot the achievable empirical risk with $50$ repetitions for each procedure.

Figure~\ref{fig: Comparison of empirical main} shows that for all tested dimensions, our sampler closely follows the exact sampler from the optimized distribution.
This implies that the proposed rejection sampling preserves the useful bias of the optimized distribution.
In contrast, sampling from the non-optimized uniform distribution consistently requires more nodes to achieve the same empirical risk in the middle regime of $N$.
For large dimension $D$, the rigorous setting $\delta = 0.1$ tends to track the exact sampler more closely.
Thus, these numerical results support the use of our classical algorithm as a replacement for exact sampling from the optimized distribution in sparse-subnetwork construction.

\begin{figure}[t]
  \centering
  \includegraphics[width=\linewidth]{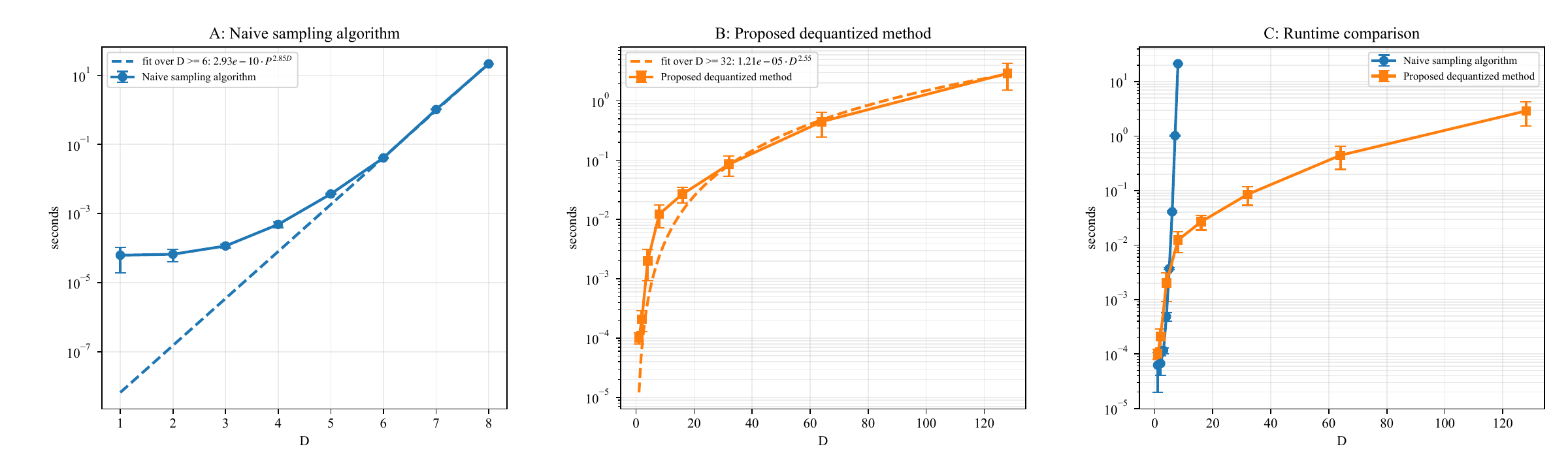}
  \caption{Runtime comparison between the naive sampling algorithm (blue solid line) and our dequantized classical algorithm (orange dashed line) in Theorem~\ref{thm: ridgelet transform}. Each curve shows the mean over $20$ executions, while error bars are 95\% confidence intervals. Fitted curves are scaling guides.}
  \label{fig: Comparison of runtimes}
\end{figure}

\paragraph{Runtime comparison.}
For the runtime experiment, we compare the runtime of our dequantized sampler with a naive classical implementation of exact sampling based on directly computing the inverse of a dense matrix.
We use the same data-generation procedure as the above empirical-risk experiment, but set $P=3$ so that the naive algorithm can be executed up to moderate dimensions $D$ beyond those in the empirical-risk experiment.
For $D\in\{1,\ldots,8\}$, the naive algorithm explicitly solves the system of linear equations over the hidden-node-parameter space before sampling.
By contrast, our sampler is tested for $D\in\{1,2,4,8,16,32,64,128\}$.

In Figure~\ref{fig: Comparison of runtimes}, the naive algorithm rapidly becomes computationally expensive as $D$ increases, with its runtime growing exponentially in $D$; in particular, for $D\geq 6$, the fitted curve scales as $P^{2.85 D}$ in the measured range.
This behavior is consistent with the cost of solving the hidden-space linear system; i.e., the number of hidden-node parameters is $P^{D + 1}$, and a dense linear solve over this space costs $O((P^{D+1})^3) = O(P^{3D})$.
In contrast, our sampler remains executable up to larger dimensions than the naive algorithm, and exhibits substantially milder growth; in particular, the fitted curve for our sampler scales as $D^{2.55}$ in the measured range.
This is consistent with the theoretical runtime bound in Theorem~\ref{thm: ridgelet transform}; that is, the dominant dimension-dependent runtime factor of our algorithm is $DM^2$ in Theorem~\ref{thm: ridgelet transform}, and in this experiment, we set $M = 50 D$, leading to the expected scale $DM^2 = O(D^3)$.
Thus, the runtime experiment numerically supports our runtime analysis: the data-dimension scaling is reduced from exponential to polynomial by our dequantized classical algorithm.

\section{Conclusion and Discussion}
We have developed a quantum-inspired classical algorithm for sampling from the optimized probability distribution used for sparse-subnetwork construction, motivated by the search for lottery tickets in dense neural networks.
In our classical algorithm, we exploit the isometry structure of the discrete ridgelet transform to avoid inverting dense matrices over the exponentially large-dimensional space by reducing it to the inversion of a diagonal matrix supported on empirical data.
This yields a polynomial-time classical algorithm when problem parameters are polynomially bounded.
The same technique also improves the previous quantum algorithm in~\citep{10.5555/3618408.3620034} for this sampling by removing this matrix-inversion requirement.
Numerical experiments show that the proposed classical sampler can construct sparse subnetworks with empirical risk comparable to exact sampling from the optimized distribution, and substantially better than sampling from the uniform distribution.
The runtime experiments further support the theoretical runtime improvement over the previous classical approach.
These results provide a classical alternative to the previous quantum approach for sparse-subnetwork construction, underscoring the importance of comparing these classical and improved quantum algorithms under plausible future hardware assumptions and hardware-specific overhead estimates.

\begin{ack}


NI was supported by JST BOOST, Japan Grant Number JPMJBS2418, JST CREST Grant Number JPMJCR25I5, JSPS KAKENHI Grant Number 23K21643, and MEXT Quantum Leap Flagship Program (MEXT QLEAP) JPMXS0118069605, JPMXS0120351339.
HY was supported by JST PRESTO Grant Number JPMJPR201A, JPMJPR23FC, JSPS KAKENHI Grant Number JP23K19970, JST CREST Grant Number JPMJCR25I5, JST [Moonshot R\&D] [Grant Number JPMJMS256J], and Faculty Research Funding from Google Quantum AI\@.
SS was supported by JST BOOST JPMJBY24E2, JST CREST Grant Number JPMJCR25I5, and JSPS KAKENHI 24K21316.
MM was supported by MEXT Quantum Leap Flagship Program (MEXT QLEAP) JPMXS0118069605, JPMXS0120351339, JST CREST Grant Number JPMJCR25I5, JST ASPIRE Grant Number JPMJAP25A3, JSPS KAKENHI Grant Number 23K21643, JST NEXUS Grant Number JPMJNX26C9, and IBM Quantum.

\end{ack}

\clearpage

\bibliographystyle{unsrtnat}
\bibliography{main}

\clearpage
\appendix

\section{Preliminaries}\label{sec:preliminaries}
\subsection{Notations and general definitions}
Let $\mathbb{N}$ be the set of natural numbers.
For $n \in \mathbb{N}$, we use $[n]$ to denote the set $\{1,\ldots,n\}$.
In this work, we abbreviate $\log_2$ to $\log$ for simplicity of notation.
Let $P$ be a prime number and $\mathbb{Z}_P \coloneq \{0,1, \ldots, P -1 \}$, where all additions and multiplications in $\mathbb{Z}_P$ are considered modulo $P$.
Unless otherwise specified, $\bm{x}$ and $\bm{a}$ range over $\mathbb{Z}_P^D$, while $b$ ranges over $\mathbb{Z}_P$.
For $\bm{a},\bm{x} \in \mathbb{Z}_P^D$, we write the inner product as $\bm{a}^\top \bm{x} \coloneq \sum_{i = 1}^D a_i x_i \pmod P$, where $a_i$ and $x_i$ denote the $i$-th elements of $\bm{a}$ and $\bm{x}$, respectively.
Let $\mb{R}$ and $\mb{C}$ be the sets of real numbers and complex numbers, respectively.
For a vector $\bm{v} \in \mb{C}^n$, $\|\bm{v}\|_2$ and $v_i$ denote the $l_2$-norm and the $i$-th entry of $\bm{v}$, respectively.
For a matrix $\bm{A}$, let $\bm{A}^\dagger$ and $\bm{A}^\top$ denote the conjugate transpose and the transpose of $\bm{A}$, respectively.
For a matrix $\bm{A} \in \mb{C}^{m\times n}$, let $\|\bm{A}\|_{\mathrm{op}}$ and $\|\bm{A}\|_{\mathrm{F}}$ denote the operator norm induced by the $l_2$-norm and the Frobenius norm, respectively.
For a Hermitian matrix $\bm{A} \in \mathbb{C}^{n \times n}$, let $\lambda_{\max}(\bm{A})$ and $\lambda_{\min}(\bm{A})$ denote the maximum and minimum eigenvalues of $\bm{A}$, respectively.

We introduce basic quantum notations (see~\cite{N4} for more introductory details).
In quantum computation, we use standard bra-ket notation for finite-dimensional Hilbert spaces.
For a finite set $\Omega$, let $\mathcal{H}_\Omega \coloneq \mathbb{C}^{|\Omega|}$.
We write this basis as $\{ \ket{\omega}  \mid \omega \in \Omega\}$.
Thus, a vector $\bm{v} \in \mathbb{C}^{|\Omega|}$ can be written as the ket
\begin{equation}
    \ket{\bm{v}} \coloneq \sum_{\omega \in \Omega} v(\omega) \ket{\omega} \in \mathcal{H}_\Omega.
\end{equation}
This ket is not assumed to be normalized unless explicitly stated.
For a ket $\ket{\bm{v}}$, its conjugate transpose is denoted by $\bra{\bm{v}}$.
For kets $\ket{\bm{v}}$ and $\ket{\bm{u}}$, their inner product is written as $\braket{\bm{v}}{\bm{u}}$.
For sets $\Omega, \Theta$ representing the standard bases, an $(\omega,\theta)$-th entry of a matrix $\bm{A} \in \mathbb{C}^{|\Omega| \times |\Theta|}$ is written as $\bra{\omega}\bm{A} \ket{\theta}$ for every $\omega \in \Omega$ and $\theta \in \Theta$.

As a conventional computation uses a bit $\{0,1\}$, quantum computation uses a quantum bit (qubit) as the minimum computational element.
A qubit is represented by $\mathbb{C}^2$, and a single-qubit state is a normalized vector in $\mathbb{C}^2$.
With respect to the computational basis $\{\ket{0},\ket{1}\}$, it has the form $\ket{\psi} = \alpha_0 \ket{0} + \alpha_1 \ket{1}$ for $\alpha_0, \alpha_1 \in \mathbb{C}$, where $|\alpha_0|^2 + |\alpha_1|^2 =1$.
An $n$-qubit register is described by the tensor-product space $(\mathbb{C}^2)^{\otimes n}$.
Its computational basis can be indexed either by bitstrings in $\{0,1\}^n$ or by integers in $\{0, 1, \ldots , 2^n -1\}$.
Therefore, any $n$-qubit state can be written as $\ket{\psi} = \sum_{z = 0}^{2^n -1} \alpha_z \ket{z}$ satisfying the normalized condition $\sum_{z}^{ 2^n - 1} |\alpha_z|^2 = 1$, and we call $\alpha_z$ an amplitude.
When two quantum registers are represented by Hilbert spaces $\mathcal{H}_\Omega$ and $\mathcal{H}_\Theta$, the composite system is represented by the tensor product $\mathcal{H}_\Omega \otimes \mathcal{H}_\Theta$.
For some basis states, we use the abbreviation $\ket{\omega, \theta} \coloneq \ket{\omega} \otimes \ket{\theta}$.
More generally, for $\bm{v} = (v_1, \ldots v_D)$, we write $\ket{\bm{v}} \coloneq \ket{v_1} \otimes  \cdots \otimes \ket{v_D}$.

A transformation of a quantum system without measurement is represented by a unitary operator.
If $\ket{\psi} \in \mathcal{H}$ is a quantum state and $\bm{U}: \mathcal{H} \to \mathcal{H}$ is unitary, then the state after the transformation is $\bm{U}\ket{\psi}$.
The unitary condition $\bm{U}^\dagger \bm{U} = \bm{U} \bm{U}^\dagger = \bm{I}$ ensures the normalization condition of quantum states, i.e. 
\begin{equation}
    \|\bm{U}\ket{\psi}\|_2^2 = \bra{\psi} \bm{U^\dagger} \bm{U} \ket{\psi} = \braket{\psi} = 1.
\end{equation}
Thus, a unitary operator maps quantum states to quantum states.
We will also use isometry transformations, which allow distinct input and output Hilbert spaces with different dimensions.
We say that a linear operator $\bm{V}: \mathcal{H}_1 \to \mathcal{H}_2$ is an isometry if $\bm{V}^\dagger \bm{V} = \bm{I}_{\mathcal{H}_1}$.
Therefore, an isometry also maps quantum states to quantum states, although the state may be embedded in a larger Hilbert space.
Although physical transformations of closed quantum systems are unitary, an isometry transformation can be implemented by a unitary transformation after adding auxiliary registers initialized in a fixed state, i.e., there exists a unitary operator $\bm{U}$ on a larger Hilbert space such that $\bm{U} (\ket{\psi} \otimes\ket{0}) = \bm{V}\ket{\psi}$ for all $\ket{\psi}$, considering $\mathcal{H}_2 \simeq \mathcal{H}_1 \otimes \mathcal{H}_\mathrm{aux}$.

A quantum measurement converts amplitudes into probability mass.
If $\ket{\psi} = \sum_{\omega \in \Omega} \alpha_\omega \ket{\omega}$ is a quantum state in $\mathcal{H}_\Omega$, then measuring it in the basis $\{\ket{\omega} \mid \omega \in \Omega\}$ outputs index $\omega \in \Omega$ with probability $|\alpha_\omega|^2$.

We say that an algorithm is a randomized classical algorithm if the classical algorithm has access to random bits.
In this paper, arithmetic (i.e., addition and multiplication) for real numbers are performed in $O(1)$ time.

In this paper, we use the total variation distance to measure the distance between two probability distributions.
\begin{definition}[Total variation distance]
    The total variation distance $d_\mathrm{TV}(\mc{D}, \mc{D}')$ between two distributions $\mc{D}$ and $\mc{D}'$ over a finite set $\mc{X}$ is defined as
    \begin{equation}
        d_\mathrm{TV}(\mc{D}, \mc{D}') = \frac{1}{2}\sum_{x\in \mc{X}} |\mc{D}(x) - \mc{D}'(x)|.
    \end{equation}
\end{definition}
It is easy to see that the total variation distance is a metric, i.e., it satisfies the triangle inequality, and $d_\mathrm{TV}(\mc{D}, \mc{D}') = 0$ if and only if $\mc{D} = \mc{D}'$, and $d_\mathrm{TV}(\mc{D}, \mc{D}') \leq 1$ for any pair of $\mc{D}$ and $\mc{D}'$.

\subsection{Quantum algorithmic subroutines}
We recall two standard subroutines used in quantum algorithms for linear algebra, i.e., state preparation and block-encoding.
These notions work as unitary transformations used in a quantum algorithm, rather than as black-box access to vectors or matrices.
\begin{definition}[State-preparation unitary]
    Let $\bm{v} \in \mathbb{C}^n$ be a nonzero vector.
    A unitary $\bm{U}$ is called a state-preparation unitary for $\bm{v}$ if it prepares the normalized amplitude encoding of $\bm{v}$,
    \begin{equation}
        \bm{U} \ket{0} = \frac{\ket{\bm{v}}}{\|\ket{\bm{v}}\|_2} \propto \sum_{i \in [n]} \frac{v_i}{\|\ket{\bm{v}}\|_2}\ket{i},
    \end{equation}
    where $\ket{0}$ denotes a fixed quantum state on a quantum register.
\end{definition}

State preparation is known as a useful QML technique since a vector is encoded into amplitudes in a quantum state.
In particular, for a probability distribution $\mathcal{D}$ over a finite set $\Omega$, the quantum state
\begin{equation}
    \ket{\sqrt{\mathcal{D}}} = \sum_{\omega \in \Omega} \sqrt{\mathcal{D}(\omega)}\ket{\omega}
\end{equation}
is normalized, and its measurement in the computational basis outputs $\omega \in \Omega$ with probability $\mathcal{D}(\omega)$, i.e., state preparation gives a quantum sampling procedure for the probability distribution $\mathcal{D}$.
Also, the state preparation is used for constructing a block-encoding of a matrix~\citep{gilyen2019quantum}.

We introduce a definition of block-encoding~\citep{gilyen2019quantum}.
Block-encoding is the standard way to represent a matrix inside a unitary operator.
\begin{definition}[Block encoding]
    Let $\bm{A} \in \mathbb{C}^{t \times t}$, $\alpha > 0$, $a \in \mathbb{N}$, and $\epsilon > 0$.
    A unitary $\bm{U}$ acting on $a + \lceil \log t\rceil$ qubits is called an $(\alpha,a,\epsilon)$-block-encoding of $\bm{A}$ if 
    \begin{equation}
        \ope{\bm{A} - \alpha
            \left(
                \bra{0}^{\otimes a}\otimes \bm{I}
            \right)
            \bm{U}
            \left(
                \ket{0}^{\otimes a}\otimes \bm{I}
            \right)}
        \leq
        \epsilon.
    \end{equation}
\end{definition}
Namely, if the unitary operator $\bm{U}$ is written in block form with respect to the auxiliary register, the block obtained by projecting the auxiliary register onto $\ket{0}^{\otimes a}$ is an approximation of $\bm{A}/\alpha$, where $\alpha$ is a normalization factor, $a$ is the number of auxiliary qubits, and $\epsilon$ is an error parameter.

Block-encoding is a particularly crucial primitive for quantum singular value transformation (QSVT) since QSVT techniques apply polynomial transformations to the singular values of a block-encoded matrix.
The quantum algorithm for sampling from the optimized probability distribution in~\citep{10.5555/3618408.3620034} uses this technique to implement the inversion of the high-dimensional matrix $\bm{R}\hat{\bm{P}}_\mathrm{data}\bm{R}^\top + \lambda P^{-D} \bm{I}$.

\subsection{Sampling and query access}
We introduce the definition of sampling and query access oracles that are widely used in quantum-inspired classical algorithms~\citep{chia2022sampling,gilyen2019quantum,le2025robust,gharibian2022dequantizing}.
The motivation for these models comes from the input primitives used in QSVT algorithms.
As discussed above, many QML algorithms begin by preparing a quantum state. Then, the algorithms apply unitary transformations such as block encodings to that quantum state.
However, such state-preparation and block-encoding unitary operations are not efficiently implemented in general, and they may require some structure of the input data to be efficiently constructed.

A typical way to implement state preparation is to build a data structure that stores partial norms of a vector in a binary tree~\citep{kerenidis_et_al:LIPIcs.ITCS.2017.49,grover2002creatingsuperpositionscorrespondefficiently}.
This type of construction allows one to prepare a quantum state whose amplitudes are proportional to the entries of the encoded vector.
On the other hand, the quantum-inspired algorithm proposed in~\cite{10.1145/3313276.3316310} suggests that the same tree structure enables us to realize classical counterparts, i.e., querying an entry, sampling an index with a probability proportional to the squared magnitude of the entry, and computing the norm of the vector.
The basic idea of quantum-inspired algorithms uses this structure to classically simulate the quantum operations.

\begin{definition}[Query access~\citep{chia2022sampling}]
    For a vector $\bm{v} \in \mb{C}^n$, let $\mathrm{Q}(\bm{v})$ be an oracle that, for any $i \in [n]$, takes as input $i \in [n]$ and outputs the value of $v_i$.
    Similarly, for a matrix $\bm{A} \in \mb{C}^{m \times n}$, let $\mathrm{Q}(\bm{A})$ be an oracle that, for any $(i,j) \in [m]\times [n]$, takes as input $(i,j)$ and outputs the value of $\bm{A}(i,j)$.
    Let $\mathbf{q}(\bm{v})$ and $\mathbf{q}(\bm{A})$ be the time complexities of one query to the oracles $\mathrm{Q}(\bm{v})$ and $\mathrm{Q}(\bm{A})$, respectively.
\end{definition}

\begin{definition}[Sampling and query access to a vector~\citep{chia2022sampling}]
    For a vector $\bm{v} \in \mb{C}^n$, let $\mathrm{SQ}(\bm{v})$ be an oracle that performs the following operations:
    \begin{itemize}
        \item For any $i \in [n]$, take as input $i$ and output the value of $v_i$.
        \item Output an index $i \in [n]$ with probability $\frac{|v_i|^2}{\|\bm{v}\|_2^2}$.
        \item Output the value of $\|\bm{v}\|_2$.
    \end{itemize}
    Let $\mathbf{q}(\bm{v}),\mathbf{s}(\bm{v})$, and $\mathbf{n}(\bm{v})$ be the time complexities of querying entries, sampling indices, and computing the norm, respectively.
    Let $\mathbf{sq}(\bm{v}) \coloneq \max(\mathbf{q}(\bm{v}), \mathbf{s}(\bm{v}), \mathbf{n}(\bm{v}))$.
\end{definition}

Block-encoding of a matrix is usually constructed from additional structures of the matrix, e.g., its sparsity and a binary tree description of the matrix.
Thus, the efficient quantum implementation of a matrix, in general, depends on a data structure or oracle that provides sufficient information about the input matrix.
Similar to accessing a vector, the sampling and query accesses to a matrix are defined using the sampling and query accesses to vectors.
\begin{definition}[Sampling and query access to a matrix~\citep{chia2022sampling}]
    For a matrix $\bm{A} \in \mb{C}^{m \times n}$ and a vector $\bm{a} \in \mathbb{R}^m$ such that $a_i \coloneq \|A_{i,\ast}\|_2$, let $\mathrm{SQ}(\bm{A})$ be an oracle that performs the following operations:
    \begin{itemize}
        \item For any $i \in [m]$, acts as $\mathrm{SQ}(A_{i,\ast})$,
        \item Acts as $\mathrm{SQ}(\bm{a})$,
    \end{itemize}
    The time complexities of querying entries, sampling rows, and computing row norms are denoted by $\mathbf{q}(\bm{A}) \coloneq \max(\mathbf{q}(A_{i,\ast}), \mathbf{q}(\bm{a}))$, $\mathbf{s}(\bm{A}) = \max(\mathbf{s}(A_{i,\ast}), \mathbf{s}(\bm{a}))$, and $\mathbf{n}(\bm{A}) = \mathbf{n}(\bm{a})$, respectively.
    Let $\mathbf{sq}(\bm{A}) \coloneq \max(\mathbf{q}(\bm{A}), \mathbf{s}(\bm{A}), \mathbf{n}(\bm{A}))$.
\end{definition}
However, constructing the block-encoding used in the quantum ridgelet sampling algorithm does not depend on the binary-tree data structure of the matrix $\bm{R}\hat{\bm{P}}_\mathrm{data}\bm{R}^\top + \lambda P^{-D} \bm{I}$, nor on the sparsity assumption.
Therefore, existing sampling-based quantum-inspired techniques proposed in~\citep{chia2022sampling,le2025robust, 10.1145/3313276.3316310} cannot be applied directly.
This motivates the need for a more problem-specific classical simulation developed in the proof of our main classical sampling theorem.

\subsection{Discrete ridgelet transform and optimized probability distribution}
We introduce the discrete ridgelet transform and its application to sparse subnetwork construction (see \citet{10.5555/3618408.3620034} for more detail).

We first define the discrete Fourier transform, which can be implemented on a quantum computer since it is a unitary transformation.
For a function $f: \mathbb{Z}_P^D \to \mathbb{R}$, define the $D$-dimensional discrete Fourier transform, for all $\bm{u} \in \mathbb{Z}_P^D$, by
\begin{equation}
    \mathcal{F}_D[f](\bm{u}) \coloneq P^{-D/2} \sum_{\bm{x}} f(\bm{x})\exp\left(- \frac{2\pi i}{P} \bm{u}^\top \bm{x}\right).
\end{equation}
Similarly, for $g: \mathbb{Z}_P \to \mathbb{R}$, define the $1$-dimensional discrete Fourier transform, for all $v \in \mathbb{Z}_P$, by
\begin{equation}
    \mathcal{F}_1[g](v) \coloneq P^{-1/2} \sum_{b}g(b) \exp\left(- \frac{2\pi i}{P} vb\right).
\end{equation}
We also assume the admissibility condition for an activation function $g$ and a ridgelet function $r: \mathbb{Z}_P \to \mathbb{R}$,
\begin{equation}
    C_{g,r} \coloneq \sum_{v} \mathcal{F}_1[g](v) \overline{\mathcal{F}_1[r](v)}\neq 0.
\end{equation}
In particular, if we choose $r = g$, then $C_{g,g} = 1$.

The following two results are the main facts about this discrete transform.
\begin{lemma}[Fourier slice theorem for discrete ridgelet transform~\citep{10.5555/3618408.3620034}]\label{lem: fourier slice}
    For any $f: \mathbb{Z}_P^D \to \mathbb{R}$, any $\bm{a} \in \mathbb{Z}_P^D$, and any $v \in \mathbb{Z}_P$, it holds that 
    \begin{equation}
        \mathcal{F}_1[\mathcal{R}[f](\bm{a},\cdot)](v) = \mathcal{F}_D[f](v \bm{a} \bmod P)\overline{\mathcal{F}_1[r](v)}.
    \end{equation}
\end{lemma}
\begin{lemma}[Exact representation of function as discretized neural network~\citep{10.5555/3618408.3620034}]\label{lem: exact represetation of discretized neural}
    For any function $f: \mathbb{Z}_P^D \to \mathbb{R}$ and any point $\bm{x} \in \mathbb{Z}_P^D$, it holds that
    \begin{equation}
        f(\bm{x}) = C_{g,r}^{-1} \mathcal{S}[\mathcal{R}[f]](\bm{x}).
    \end{equation}
\end{lemma}
Lemma~\ref{lem: exact represetation of discretized neural} gives us considering $\mathcal{R}[f](\bm{a},b)$ as a coefficient assigned to the hidden node parameterized by $(\bm{a},b)$.
Since the representation is overcomplete, the coefficients are not unique in general.
Nevertheless, this gives us explicit coefficients that reconstruct the function $f$ exactly on the finite domain $\mathbb{Z}_P^D$.

The discrete ridgelet transform $\bm{R}: \mathbb{R}^{P^D} \to \mathbb{R}^{P^{D + 1}}$ is an isometry, i.e., 
\begin{equation}
    \bm{R}^\top \bm{R} = \bm{I}_{P^D}.
\end{equation}
Therefore, there exists a quantum implementation of $\bm{R}$ on the quantum system.
The quantum ridgelet transform is the quantum algorithm that implements this isometry efficiently by using $\ket{r} = \sum_b r(b) \ket{b}$, where $\ket{r}$ and $\mathrm{SQ}(r)$ can be prepared efficiently in runtime $O(\polylog(P))$.
Using the Fourier slice theorem in Lemma~\ref{lem: fourier slice}, and quantum Fourier transform over $\mathbb{Z}_P$, the quantum ridgelet transform performs
\begin{equation}
    \ket{\psi} \mapsto \bm{R}\ket{\psi}
\end{equation}
within runtime $O(D\, \polylog (P))$, where $\ket{\psi}$ is a quantum state on the input-domain Hilbert space $\mathbb{C}^{P^D}$.

We summarize how the optimized probability distribution in Definition~\ref{def: optimized probabilit distributon} used for sparse-subnetwork construction is defined.
For readability, we briefly restate the setting of sampling.
As mentioned in the manuscript, we take $r = g$.
Let $\{\bm{x}_m, y_m\}_{m \in [M]}$ be a training dataset of $M$ pairs of samples with $y_m = f(\bm{x}_m)$.
Let $\hat{p}_\mathrm{data}$ be the empirical distribution of the input data and define
\begin{equation}
    \hat{\bm{P}}_\mathrm{data} \coloneq \sum_{\bm{x}} \hat{p}_\mathrm{data}(\bm{x}) \ketbra{\bm{x}}.
\end{equation}
We also write
\begin{equation}
    \ket{\psi_\mathrm{in}} \coloneq \hat{\bm{P}}_\mathrm{data} \ket{f} = \sum_{\bm{x}}\hat{p}_\mathrm{data}(\bm{x}) f(\bm{x}) \ket{\bm{x}},
\end{equation}
where this vector is not necessarily normalized.
Consider the regularized weighted least-squares regression to be minimized in~\eqref{eq : optimal estimator}
\begin{equation}
    \tilde{J}(w) = \sum_{\bm{x}} \hat{p}_\mathrm{data}(\bm{x}) |f(\bm{x}) - \mathcal{S}[w](\bm{x})|^2 + \lambda \sum_{\bm{a},b} \left| P^{-D/2} w(\bm{a},b)\right|^2.
\end{equation}
This can be written by using the matrix notation as
\begin{equation}
    \tilde{J}(w) = \|\hat{\bm{P}}_\mathrm{data}^{1/2}(\ket{f} - \bm{R}^\top \ket{w})\|_2^2 + \lambda P^{-D} \| \ket{w}\|_2^2,
\end{equation}
where $\ket{w} = \sum_{\bm{a},b} w(\bm{a},b) \ket{\bm{a},b}$.
The regularization term makes the minimization of $ \tilde{J}(w)$ strongly convex in $\ket{w}$.
By differentiating this quadratic function in $\ket{w}$, we obtain the following equation
\begin{equation}
    \left(\bm{R}\hat{\bm{P}}_\mathrm{data}\bm{R}^\top + \lambda P^{-D} \bm{I}\right) \ket{w_\lambda^\ast} = \bm{R} \hat{\bm{P}}_\mathrm{data} \ket{f} = \bm{R} \ket{\psi_\mathrm{in}}.
\end{equation}
Therefore, the optimal solution vector is given by 
\begin{equation}\label{eq: optimal solution ket form}
    \ket{w_\lambda^\ast} = \left(\bm{R}\hat{\bm{P}}_\mathrm{data}\bm{R}^\top + \lambda P^{-D} \bm{I}\right)^{-1} \bm{R} \ket{\psi_\mathrm{in}}.
\end{equation}
Namely, the coefficient on the hidden node parameterized by $(\bm{a},b)$ is $w_\lambda^\ast(\bm{a},b) = \bra{\bm{a},b} \ket{w_\lambda^\ast} $.

In the manuscript, the optimized probability distribution $p_{\lambda,\Delta}^\ast$ is defined as 
\begin{equation}
    p_{\lambda,\Delta}^\ast(\bm{a},b)
    \coloneq
    \frac{1}{Z}
    \frac{|P^{-D/2}w_\lambda^\ast(\bm{a},b)|^2}
         {|P^{-D/2}w_\lambda^\ast(\bm{a},b)|^2 + \Delta},
    \qquad
    \Delta>0,
\end{equation}
where
\begin{equation}
    Z \coloneq \sum_{\bm{a},b} \frac{|P^{-D/2}w_\lambda^\ast(\bm{a},b)|^2}{|P^{-D/2}w_\lambda^\ast(\bm{a},b)|^2 + \Delta}
\end{equation}
is the normalization constant, and the parameter $\Delta$ smooths the distribution.
The corresponding quantum sampling state is
\begin{equation}
    \ket{p_{\lambda,\Delta}^\ast} = \frac{1}{\sqrt{Z}} \sum_{\bm{a},b}
    \frac{P^{-D/2}w_\lambda^\ast(\bm{a},b)}
         {\sqrt{|P^{-D/2}w_\lambda^\ast(\bm{a},b)|^2 + \Delta}} \ket{\bm{a},b}.
\end{equation}
To represent this in matrix form, we define a diagonal matrix $\bm{W}_\lambda$ by
\begin{equation}\label{eq: W lambda def}
    \bm{W}_\lambda \coloneq \frac{1}{\gamma} \sum_{\bm{a},b} |P^{-D/2}w_\lambda^\ast(\bm{a},b)|^2 \ketbra{\bm{a},b},
\end{equation}
where $\gamma$ is defined in~\eqref{eq: definition gamma}.
Then, we can write a vector proportional to $\ket{p_{\lambda,\Delta}^\ast}$ as 
\begin{equation}\label{eq: optimized probability distoribution state}
    \ket{p_{\lambda,\Delta}^\ast} \propto \left(\bm{W}_\lambda + \frac{\Delta}{\gamma}\bm{I}\right)^{-1/2} \left(\bm{R}\hat{\bm{P}}_\mathrm{data}\bm{R}^\top + \lambda P^{-D} \bm{I}\right)^{-1} \bm{R} \ket{\psi_\mathrm{in}}.
\end{equation}
Measuring this state in the computational basis outputs $(\bm{a},b)$ according to $p_{\lambda,\Delta}^\ast$.
Thus, the quantum algorithm constructs a sparse subnetwork by repeatedly preparing this state and measuring it.
The quantum ridgelet-transform algorithm prepares this state by combining state preparation, the quantum ridgelet transform, block encodings, and QSVT~\citep{10.5555/3618408.3620034}.

\section{Key technical lemma}~\label{sec: key technical lemma}
In this section, we prove the identity that is used throughout the proof of both our classical and improved quantum sampling theorems.
The computational difficulty in sampling from the optimized probability distribution $p_{\lambda,\Delta}^\ast$ is the inverse
\begin{equation}
    \left(\bm{R}\hat{\bm{P}}_\mathrm{data}\bm{R}^\top + \lambda P^{-D} \bm{I}\right)^{-1}.
\end{equation}
This acts on the parameter space of dimension $P^{D+1}$, and the matrix is not sparse in the standard basis.
The following lemma shows that this inverse has a simple form by using the property of $\bm{R}$, i.e., $\bm{R}$ is an isometry.
\begin{lemma}[Inverse decomposition]\label{lem: inverse decompotion}
    \begin{equation}\label{eq: inverse decomposition}
        \left(\bm{R}\hat{\bm{P}}_\mathrm{data}\bm{R}^\top + \lambda P^{-D} \bm{I}\right)^{-1} = \bm{R} \left(\hat{\bm{P}}_\mathrm{data} + \lambda P^{-D} \bm{I}\right)^{-1} \bm{R}^\top + \frac{1}{\lambda P^{-D}}(\bm{I} - \bm{R}\bm{R}^\top)
    \end{equation}
\end{lemma}
\begin{proof}
    We can check the equality by applying the right-hand side to the original matrix, i.e.,
    \begin{equation}
    \begin{aligned}
        &\quad \left\{\bm{R} \left(\hat{\bm{P}}_\mathrm{data} + \lambda P^{-D} \bm{I}\right)^{-1} \bm{R}^\top + \frac{1}{\lambda P^{-D}}(\bm{I} - \bm{R}\bm{R}^\top)\right\} \left(\bm{R}\hat{\bm{P}}_\mathrm{data}\bm{R}^\top + \lambda P^{-D} \bm{I}\right) \\
        &= \bm{R} \left(\hat{\bm{P}}_\mathrm{data} + \lambda P^{-D} \bm{I}\right)^{-1} \hat{\bm{P}}_\mathrm{data} \bm{R}^\top + \lambda P^{-D} \bm{R} \left(\hat{\bm{P}}_\mathrm{data} + \lambda P^{-D} \bm{I}\right)^{-1} \bm{R}^\top + (\bm{I} - \bm{R}\bm{R}^\top) \\
        & = \bm{R} \left(\hat{\bm{P}}_\mathrm{data} + \lambda P^{-D} \bm{I}\right) \left(\hat{\bm{P}}_\mathrm{data} + \lambda P^{-D} \bm{I}\right)^{-1} \bm{R}^\top + (\bm{I} - \bm{R}\bm{R}^\top) \\
        & = \bm{I}.
    \end{aligned}
    \end{equation}
    Since $\bm{R}\hat{\bm{P}}_\mathrm{data}\bm{R}^\top + \lambda P^{-D} \bm{I}$ is positive definite, due to $\lambda > 0$, its inverse exists, which yields the conclusion.
\end{proof}

The key consequence from this lemma is that the second term in~\eqref{eq: inverse decomposition} vanishes when the inverse is applied to a vector in the image of $\bm{R}$, i.e.,
\begin{equation}
    (\bm{I} - \bm{R}\bm{R}^\top) \bm{R} = 0.
\end{equation}
Therefore, the optimal solution $\ket{w_\lambda^\ast}$ in~\eqref{eq: optimal solution ket form} is reduced to
\begin{equation}\label{eq: reduced optimal solution}
\begin{aligned}
    \ket{w_\lambda^\ast} &= \left(\bm{R}\hat{\bm{P}}_\mathrm{data}\bm{R}^\top + \lambda P^{-D} \bm{I}\right)^{-1} \bm{R} \ket{\psi_\mathrm{in}} \\
    & = \bm{R} \left(\hat{\bm{P}}_\mathrm{data} + \lambda P^{-D} \bm{I}\right)^{-1} \ket{\psi_\mathrm{in}}.
\end{aligned}
\end{equation}

\section{Proof of Theorem~\ref{thm: ridgelet transform}}\label{sec: proof of ridgelet}
In this Appendix, we give the proof of Theorem~\ref{thm: ridgelet transform} by using Algorithm~\ref{alg: ridgelet sampling} with Lemma~\ref{lem: inverse decompotion}.

\begin{algorithm}[tbp]
\SetAlgoLined
\caption{Classical sampling from the optimized probability distribution~\label{alg: ridgelet sampling}}
\KwIn{precision parameters $\lambda, \Delta > 0$, $\delta \in (0,1)$, and $\gamma$ in~\eqref{eq: definition gamma}, and the data set $\{\bm{x}_m, y_m\}_{m \in [M]}$.}
\KwOut{A random variable $(\bm{a},b) \in \mathbb{Z}_P^D \times \mathbb{Z}_P$ sampled from a probability distribution ${\mc{D}}'$ close to the distribution $p_{\lambda,\Delta}^\ast$ in Definition~\ref{def: optimized probabilit distributon} at most $\delta$, i.e., $\dtv(\mc{D}', p_{\lambda,\Delta}^\ast) \leq \delta$.}
Let $K$ be the number of distinct $\bm{x}_m$'s, i.e. $K = |\{\bm{x}_m\mid m \in [M]\}|$\;
Set
\begin{equation}
    I = \left\lceil K \left(1 + \frac{\gamma}{\Delta}\right)\ln(1/\delta)\right\rceil
\end{equation}
and $A = 0$\;
Construct a data structure to implement $\mathrm{SQ}(\ket{\phi})$ to $\ket{\phi}$ in~\eqref{eq: definition phi ridgelet}\;
\For{$i \in [I]$}{
    Sample $\bm{x} \in \mathbb{Z}_P^D$ from the distribution proportional to $\ket{\phi}$ by $\mathrm{SQ}(\ket{\phi})$\;
    Sample $\bm{a} \in \mathbb{Z}_P^D$ from the uniform distribution over $\mathbb{Z}_P^D$\;
    Sample $t \in \mathbb{Z}_P$ with probability $|r(t)|^2$ with $\mathrm{SQ}(r)$, and let $b = \bm{a}^\top \bm{x} - t$\;
    Compute $ \frac{\Delta}{\Delta + s(\bm{a},b)}\frac{s(\bm{a},b)}{K\gamma q(\bm{a},b)}$ in~\eqref{eq: thining rate ridgelet};
    Sample $u$ from the uniform distribution over $[0,1]$\;
    \If{$u \leq \frac{\Delta}{\Delta + s(\bm{a},b)}\frac{s(\bm{a},b)}{K\gamma q(\bm{a},b)}$}{Set $A = 1$\;}
    \If{$A = 1$}{\Break\;}}
\If{$A = 0$}{Sample $(\bm{a},b)$ uniformly from $\mathbb{Z}_P^D \times \mathbb{Z}_P$\;}
\Return $(\bm{a},b)$\;
\end{algorithm}

\begin{theorem}[The restatement of Theorem~\ref{thm: ridgelet transform}\label{thm: righelet restatement}]
        Given parameters $\lambda,\Delta > 0$, $\delta \in (0,1)$, and $\gamma$ in~\eqref{eq: definition gamma}, and the dataset $\{\bm{x}_m, y_m\}_{m \in [M]}$, Algorithm~\ref{alg: ridgelet sampling} outputs one sample from a distribution $\mathcal{D}'$ whose total variation distance from $p_{\lambda,\Delta}^\ast$ in Definition~\ref{def: optimized probabilit distributon} is at most $\delta$, i.e.,
    \begin{equation}
        \dtv(p_{\lambda,\Delta}^\ast,\mathcal{D}') \leq \delta,
    \end{equation}
    with the sample complexity 
    \begin{equation}
        \tilde{O}\left(\frac{DM^2 \gamma}{\Delta} \log\frac{1}{\delta}\right).
    \end{equation}
\end{theorem}
\begin{proof}
    We first show the correctness of Algorithm~\ref{alg: ridgelet sampling}, and then we analyze the runtime.

    Let $\mathcal{X}_{\mathrm{data}}$ be the support of the empirical distribution $\hat{p}_\mathrm{data}$, i.e. $\mathcal{X}_{\mathrm{data}} = \{ x_m \mid m \in [M]\}$ and let $K \coloneq |\mathcal{X}_{\mathrm{data}}| \leq M$.
    For simplicity of notation, let us write $s(\bm{a},b) = |P^{-D/2}w^{\ast}_\lambda(\bm{a},b)|^2$, i.e.,
    \begin{equation}\label{eq: definition s(a,b)}
    \begin{aligned}
        s(\bm{a},b) &= \left|\bra{\bm{a},b} P^{-D/2}\bm{R} (\hat{\bm{P}}_{\mathrm{data}} + \lambda P^{-D} \bm{I})^{-1} \hat{\bm{P}}_{\mathrm{data}} \ket{f}\right|^2, \\
        &= \left| \sum_{\bm{x}} P^{-D}r\left((\bm{a}^\top \bm{x} - b )\bmod P \right) \frac{\hat{p}_{\mathrm{data}}(\bm{x})f(\bm{x})}{\lambda P^{-D} + \hat{p}_{\mathrm{data}}(\bm{x})} \right|^2,
    \end{aligned}
    \end{equation}
    from~\eqref{eq: reduced optimal solution}, and thus it holds that
    \begin{equation}\label{eq: sum s(a,b)}
    \begin{aligned}
        &\sum_{\bm{a},b} s(\bm{a},b) \\
        =& P^{-D}\sum_{\bm{a},b} \bra{f}\hat{\bm{P}}_{\mathrm{data}} (\hat{\bm{P}}_{\mathrm{data}} + \lambda P^{-D} \bm{I})^{-1} \bm{R}^\top \ket{\bm{a},b}\bra{\bm{a},b} \bm{R}(\hat{\bm{P}}_{\mathrm{data}} + \lambda P^{-D} \bm{I})^{-1} \hat{\bm{P}}_{\mathrm{data}} \ket{f} \\
        =& P^{-D} \bra{f}\hat{\bm{P}}_{\mathrm{data}} (\hat{\bm{P}}_{\mathrm{data}} + \lambda P^{-D} \bm{I})^{-1} \bm{R}^\top \bm{I} \bm{R}(\hat{\bm{P}}_{\mathrm{data}} + \lambda P^{-D} \bm{I})^{-1} \hat{\bm{P}}_{\mathrm{data}} \ket{f} \\
        =& P^{-D} \bra{f}\hat{\bm{P}}_{\mathrm{data}} (\hat{\bm{P}}_{\mathrm{data}} + \lambda P^{-D} \bm{I})^{-1} (\hat{\bm{P}}_{\mathrm{data}} + \lambda P^{-D} \bm{I})^{-1} \hat{\bm{P}}_{\mathrm{data}} \ket{f} \\
        =& \|P^{-D/2}(\hat{\bm{P}}_{\mathrm{data}} + \lambda P^{-D} \bm{I})^{-1} \hat{\bm{P}}_{\mathrm{data}} \ket{f}\|_2^2.
    \end{aligned}
    \end{equation}

    Algorithm~\ref{alg: ridgelet sampling} tries to sample from the distribution $p_{\lambda,\Delta}^\ast(\bm{a},b)$ in Definition~\ref{def: optimized probabilit distributon} by rejection sampling.
    To achieve this, we construct a data structure to implement sampling and query access oracle $\mathrm{SQ}(\ket{\phi})$ to the vector
    \begin{equation}\label{eq: definition phi ridgelet}
        \begin{aligned}
            \ket{\phi} &\coloneq P^{-D/2}(\hat{\bm{P}}_{\mathrm{data}} + \lambda P^{-D} \bm{I})^{-1} \hat{\bm{P}}_{\mathrm{data}} \ket{f} \\
            &= P^{-D/2} \sum_{\bm{x}} \frac{\hat{p}_{\mathrm{data}}(\bm{x})f(\bm{x})}{\lambda P^{-D} + \hat{p}_{\mathrm{data}}(\bm{x})} \ket{\bm{x}},
        \end{aligned}
    \end{equation}
    which can be achieved within $\tilde{O}(M)$ time by constructing the binary-tree structure used in~\citep{chia2022sampling,kerenidis_et_al:LIPIcs.ITCS.2017.49, 10.5555/3618408.3620034}, where the query complexity 
    \begin{equation}
        \mathbf{sq}(\ket{\phi}) = \polylog(M),
    \end{equation}
    and $\|\ket{\phi}\|_2^2 = \gamma$ from~\eqref{eq: definition gamma}.
    
    Then, we define the proposal probability distribution $q$ over $\mathbb{Z}_P^D \times \mathbb{Z}_P$ for the rejection sampling from $p_{\lambda,\Delta}^\ast$ as follows,
    \begin{equation}
        q(\bm{a},b) \coloneq \frac{1}{\gamma} \sum_{\bm{x}} \left|P^{-D/2}\frac{\hat{p}_{\mathrm{data}}(\bm{x})f(\bm{x})}{\lambda P^{-D} + \hat{p}_{\mathrm{data}}(\bm{x})}\right|^2\left|P^{-D/2}r\left((\bm{a}^\top \bm{x} - b )\bmod P \right)\right|^2.
    \end{equation}
    Algorithm~\ref{alg: ridgelet sampling} samples from this distribution $q$ by sampling $\bm{x} \in \mathbb{Z}_P^D$ from $\ket{\phi}$ with $\mathrm{SQ}(\ket{\phi})$, $\bm{a} \in \mathbb{Z}_P^D$ uniformly, and $t \in \mathbb{Z}_P$ with probability $|r(t)|^2$ with $\mathrm{SQ}(r)$, and letting $b = \bm{a}^\top \bm{x} - t$.
    The probability of obtaining $(\bm{a},b)$ in this process holds
    \begin{equation}
        \begin{aligned}
            \pr(\bm{a},b) &= \sum_{\bm{x}} \pr(\bm{x},\bm{a},b) \\
            &= \sum_{\bm{x}} \pr(\bm{x},\bm{a},t) \\
            &= \sum_{\bm{x}} \pr(\bm{x}) \pr(\bm{a}) \pr(t)\\
            &= \sum_{\bm{x}} \frac{\left|P^{-D/2}\frac{\hat{p}_{\mathrm{data}}(\bm{x})f(\bm{x})}{\lambda P^{-D} + \hat{p}_{\mathrm{data}}(\bm{x})}\right|^2}{\|\ket{\phi}\|_2^2} \frac{1}{P^{D}} |r(t)|^2 \\
            &= \sum_{\bm{x}} \frac{\left|P^{-D/2}\frac{\hat{p}_{\mathrm{data}}(\bm{x})f(\bm{x})}{\lambda P^{-D} + \hat{p}_{\mathrm{data}}(\bm{x})}\right|^2}{\|\ket{\phi}\|_2^2} \frac{1}{P^{D}} |r(\bm{a}^\top \bm{x} - b)|^2 \\
            & = q(\bm{a},b).
        \end{aligned}
    \end{equation}
    Therefore, Algorithm~\ref{alg: ridgelet sampling} exactly samples from $q$.
    
    Then, we accept $(\bm{a},b)$ sampled from the probability distribution $q$ with probability
    \begin{equation}\label{eq: thining rate ridgelet}
        \frac{\Delta}{\Delta + s(\bm{a},b)}\frac{s(\bm{a},b)}{K\gamma q(\bm{a},b)} \leq 1,
    \end{equation}
    where the inequality holds from the Cauchy-Schwarz inequality.
    The probability $\pr(\bm{a},b: \text{accept})$ of each $(\bm{a},b) \in \mathbb{Z}_P^D \times \mathbb{Z}_P$ is returned is given by 
    \begin{equation}
    \begin{aligned}
        \pr(\bm{a},b: \text{accept}) &= \frac{\Delta}{\Delta + s(\bm{a},b)} \frac{s(\bm{a},b)}{K\gamma q(\bm{a},b)} q(\bm{a},b)\\
        &= \frac{\Delta}{K\gamma} \frac{s(\bm{a},b)}{\Delta + s(\bm{a},b)}.
    \end{aligned}
    \end{equation}
    Thus, the probability that a certain $(\bm a,b) \in \mathbb{Z}_P^D \times \mathbb{Z}_P$ is accepted in Algorithm~\ref{alg: ridgelet sampling} is
    \begin{equation}
        \begin{aligned}
            \pr(\text{accept}) &= \sum_{\bm{a},b} \pr(\bm{a},b: \text{accept})\\
            & = \frac{\Delta}{K\gamma} \sum_{\bm{a},b} \frac{s(\bm{a},b)}{\Delta + s(\bm{a},b)},
        \end{aligned}
    \end{equation}
    Notice that the lower bound of $\pr(\text{accept})$ is obtained as follows
    \begin{equation}\label{eq: lower bound of pr accept s(a,b)}
        \begin{aligned}
            \pr(\text{accept}) & = \frac{\Delta}{K\gamma} \sum_{\bm{a},b} \frac{s(\bm{a},b)}{\Delta + s(\bm{a},b)} \\
            & \geq \frac{\Delta}{K\gamma} \sum_{\bm{a},b} \frac{s(\bm{a},b)}{\Delta + \gamma} \\
            & = \frac{\Delta}{K\gamma} \frac{\gamma}{\Delta + \gamma} \\
            & = \frac{1}{K} \frac{\Delta}{\Delta + \gamma},
        \end{aligned}
    \end{equation}
    where we use the fact that $s(\bm{a},b) \leq \gamma$ in the inequality.
    Therefore, the probability distribution returned by Algorithm~\ref{alg: ridgelet sampling} when the rejection sampling succeeds is given by 
    \begin{equation}\label{eq: probability distribution s(a,b)}
    \begin{aligned}
        \pr(\bm{a},b) &= \frac{\pr(\bm{a},b: \text{accept})}{\pr(\text{accept})} \\
        &= \frac{s(\bm{a},b)}{\Delta + s(\bm{a},b)} \left(\sum_{\bm{a},b} \frac{s(\bm{a},b)}{\Delta + s(\bm{a},b)}\right)^{-1} \\
        &= p_{\lambda,\Delta}^\ast(\bm{a},b).
    \end{aligned}
    \end{equation}
    Algorithm~\ref{alg: ridgelet sampling} repeats this sampling for $I$ times, and thus, the probability that the rejection sampling fails is given by
    \begin{equation}
        \begin{aligned}
            \pr(\text{failure}) &= (1 - \pr(\text{accept}))^I \\
            & \leq \exp(-I \times \pr(\text{accept})) \\
            & \leq \exp(- I \times \frac{1}{K} \frac{\Delta}{\Delta + \gamma}) \\
            &= \exp(- \left\lceil K \left(1 + \frac{\gamma}{\Delta}\right)\ln(1/\delta)\right\rceil \times \frac{1}{K} \frac{\Delta}{\Delta + \gamma}) \\
            &\leq \exp(- K \left(1 + \frac{\gamma}{\Delta} \right)\ln(1/\delta) \times \frac{1}{K} \frac{\Delta}{\Delta + \gamma}) \\
            &\leq \delta.
        \end{aligned}
    \end{equation}
    When the sampling fails, Algorithm~\ref{alg: ridgelet sampling} outputs $(\bm{a},b) \in \mathbb{Z}_P^D \times \mathbb{Z}_P$ sampled uniformly.
    This completes the proof of correctness in Algorithm~\ref{alg: ridgelet sampling}.

    Finally, we analyze the runtime of Algorithm~\ref{alg: ridgelet sampling}.
    The runtime of pre-processing is dominated by constructing the data structure of $\mathrm{SQ}(\ket{\phi})$, which takes $\tilde{O}(M)$.
    The runtime of Algorithm~\ref{alg: ridgelet sampling} is longest when the algorithm repeats $I$ times, and that fails.
    In this case, the runtime is dominated by the runtime of $I$ computations of the acceptance rate in Equation~\eqref{eq: thining rate ridgelet}.
    The dominant part of computing this value is to compute $s(\bm{a},b)$ in Equation~\eqref{eq: definition s(a,b)}.
    To compute this value, we compute inner products $\bm{a}^\top \bm{x}$, query $\mathrm{SQ}(r)$ and $\mathrm{SQ}(\ket{\phi})$ $K$ times and sum of the $K$ values, i.e., it takes $O(K (D + \mathbf{sq}(r) + \mathbf{sq}(\ket{\phi})))$.
    Therefore, the total runtime of Algorithm~\ref{alg: ridgelet sampling} is upper bounded by 
    \begin{equation}
            I \times K (D + \mathbf{sq}(r) + \mathbf{sq}(\ket{\phi})) = O(D\log(1/\delta)) \times \tilde{O}\left(\frac{M^2 \gamma}{\Delta} \right),
    \end{equation}
    which yields the conclusion of Theorem~\ref{thm: righelet restatement}.
\end{proof}

\section{Proof of Theorem~\ref{thm: improved quantum sampler in main}}\label{sec: proof of quantum sampler}

\begin{algorithm}[tbp]
\SetAlgoLined
\caption{Quantum sampling from the optimized distribution~\label{alg: quantum algorithm}}
\KwIn{precision parameters $\lambda, \Delta, \delta > 0$ and $\gamma$ in~\eqref{eq: definition gamma}, and the data set $\{\bm{x}_m, y_m\}_{m \in [M]}$.}
\KwOut{A random variable $(\bm{a},b) \in \mathbb{Z}_P^D \times \mathbb{Z}_P$ sampled from a probability distribution ${\mc{D}}'$ close to the distribution $p_{\lambda,\Delta}^\ast$ in Definition~\ref{def: optimized probabilit distributon} at most $\delta$, i.e., $\dtv(\mc{D}', p_{\lambda,\Delta}^\ast) \leq \delta$.}
Prepare a quantum state $\ket{\hat{\phi}} \propto \sum_{\bm{x}}  \frac{\hat{p}_{\mathrm{data}}(\bm{x})f(\bm{x})}{\lambda P^{-D} + \hat{p}_{\mathrm{data}}(\bm{x})} \ket{\bm{x}}$\;
Apply quantum ridgelet transform to obtain a state $\bm{R}\ket{\hat{\phi}}$\;
Apply $(\bm{W}_\lambda+ \frac{\Delta}{\gamma}\bm{I})^{-1/2}$ by QSVT to obtain
\begin{equation}
    \ket{p_{\lambda,\Delta}^\ast} \propto \left(\bm{W}_\lambda + \frac{\Delta}{\gamma}\bm{I}\right)^{-1/2} \left(\bm{R}\hat{\bm{P}}_\mathrm{data}\bm{R}^\top + \lambda P^{-D} \bm{I}\right)^{-1} \bm{R} \ket{\psi_\mathrm{in}}.\;
\end{equation}
Perform a quantum measurement in the standard basis to sample $(\bm{a},b)$ with the probability $p_{\lambda,\Delta}^\ast(\bm{a},b)$.\;
\Return $(\bm{a},b)$\;
\end{algorithm}

\begin{theorem}[The restatement of Theorem~\ref{thm: improved quantum sampler in main}]
\label{thm: improved quantum sampler}
Given parameters $\lambda,\Delta,\delta>0$, $\gamma$ in~\eqref{eq: definition gamma}, and the dataset $\{(\bm{x}_m,y_m)\}_{m\in[M]}$, Algorithm~\ref{alg: quantum algorithm} outputs one sample from a distribution $\mathcal D'$ satisfying
\begin{equation}
    \dtv(\mathcal D',p_{\lambda,\Delta}^\ast)\leq \delta
\end{equation}
with the sample complexity
\begin{equation}
    \widetilde O\!\left(
        \frac{D\gamma}{\Delta}
        \polylog\frac{MP}{\delta}
    \right).
\end{equation}
after preprocessing $\tilde{O}(M)$.
\end{theorem}
\begin{proof}
    Algorithm~\ref{alg: quantum algorithm} basically follows the quantum algorithm presented in~\citep{10.5555/3618408.3620034}, but uses Lemma~\ref{lem: inverse decompotion} to avoid implementing the inverse
    \begin{equation}
        \left(\bm{R}\hat{\bm{P}}_\mathrm{data}\bm{R}^\top + \lambda P^{-D} \bm{I}\right)^{-1}
    \end{equation}
    by QSVT.
    We first analyze the correctness of Algorithm~\ref{alg: quantum algorithm}, and then summarize its total runtime.
    
    By Lemma~\ref{lem: inverse decompotion}, we have the following representation of the optimal solution in~\eqref{eq : optimal estimator}
    \begin{equation}\label{eq: ket w reduced}
        \ket{w_\lambda^\ast} = \bm{R} \left(\hat{\bm{P}}_\mathrm{data} + \lambda P^{-D} \bm{I}\right)^{-1} \hat{\bm{P}}_\mathrm{data} \ket{f}.
    \end{equation}
    Similar to Algorithm~\ref{alg: ridgelet sampling}, we define a vector
    \begin{equation}\label{eq: phi def in quantum}
        \ket{\phi}= P^{-D/2} \left(\hat{\bm{P}}_\mathrm{data} + \lambda P^{-D} \bm{I}\right)^{-1} \hat{\bm{P}}_\mathrm{data} \ket{f},
    \end{equation}
    and since $\hat{\bm{P}}_\mathrm{data}$ is diagonal, we have the explicit form
    \begin{equation}
        \ket{\phi}= P^{-D/2} \sum_{\bm{x}}  \frac{\hat{p}_{\mathrm{data}}(\bm{x})f(\bm{x})}{\lambda P^{-D}+ \hat{p}_{\mathrm{data}}(\bm{x})} \ket{\bm{x}}.
    \end{equation}
    In particular, the support of $\hat{\bm{P}}_\mathrm{data}$ is at most $O(M)$ and we can evaluate each value in constant time.
    Therefore, as in the original quantum sampling in~\citep{10.5555/3618408.3620034}, Algorithm~\ref{alg: quantum algorithm} prepares the quantum state proportional to $\ket{w_\lambda^\ast}$ within $\tilde{O}(\polylog(M))$ by constructing the binary-tree structure~\cite{grover2002creatingsuperpositionscorrespondefficiently, kerenidis_et_al:LIPIcs.ITCS.2017.49} which takes $O(M)$.
    Namely, we can construct a unitary operation $\bm{U}$ such that
    \begin{equation}\label{eq: unitary purification}
        \bm{U}\ket{0} = \frac{\ket{\phi}}{\|\ket{\phi}\|_2} \eqcolon \ket{\hat{\phi}}
    \end{equation}
    implemented by a $\polylog(M)$-depth quantum circuit after $\tilde{O}(M)$ preprocessing.

    By~\eqref{eq: ket w reduced} and~\eqref{eq: phi def in quantum}, we have
    \begin{equation}
        \ket{w_\lambda^\ast} = P^{D/2} \bm{R} \ket{\phi}.
    \end{equation}
    Thus, after preparing $\ket{\hat{\phi}}$ and applying the quantum ridgelet transform, Algorithm~\ref{alg: quantum algorithm} obtains
    \begin{equation}
        \bm{R}\ket{\hat{\phi}} = \frac{1}{\sqrt{\gamma}} \sum_{\bm{a},b} P^{-D/2} w_\lambda^\ast(\bm{a},b) \ket{\bm{a},b},
    \end{equation}
    which takes ${O}(D\;\polylog(MP))$ runtime.

    We construct the block-encoding of $\bm{W}_\lambda+ \frac{\Delta}{\gamma}\bm{I}$, where $\bm{W}_\lambda$ in~\eqref{eq: W lambda def}.
    This can be done by using the block-encoding unitaries of $\bm{W}_\lambda$ and $\bm{I}$.
    The block-encoding of $I$ is a trivial unitary operator, and thus, consider the block-encoding of $\bm{W}_\lambda$.
    The unitary operator $\bm{U}$ in~\eqref{eq: unitary purification} and the quantum ridgelet transform $\bm{R}$ produce the purification of $\bm{W}_\lambda$.
    Namely, adding an auxiliary register and applying CNOT to $\bm{R}\ket{\hat{\phi}}$, we obtain a quantum state
    \begin{equation}
        \frac{1}{\sqrt{\gamma}} \sum_{\bm{a},b} P^{-D/2} w_\lambda^\ast(\bm{a},b) \ket{\bm{a},b} \otimes \ket{\bm{a},b}.
    \end{equation}
    By tracing out the second register, this exactly yields
    \begin{equation}
        \bm{W}_\lambda \coloneq \frac{1}{\gamma} \sum_{\bm{a},b} |P^{-D/2}w_\lambda^\ast(\bm{a},b)|^2 \ketbra{\bm{a},b},
    \end{equation}
    then we can construct a $(1, O((D + 1)\log P), \delta)$-block encoding of $\bm{W}_\lambda$ by using $\bm{R},\bm{R}^\dagger, \bm{U}$, and $\bm{U}^\dagger$ a constant number of times~\citep{gilyen2019quantum,10.5555/3618408.3620034}.
    The construction of $\bm{W}_\lambda+ \frac{\Delta}{\gamma}\bm{I}$ can be implemented by using the block-encoding of $\bm{W}_\lambda$ and $\bm{I}$ a constant number of times.

    To apply $(\bm{W}_\lambda+ \frac{\Delta}{\gamma}\bm{I})^{-1/2}$, we use QSVT of the block-encoding of $\bm{W}_\lambda+ \frac{\Delta}{\gamma}\bm{I}$, which can be implemented by using the block-encoding and its conjugate transpose linear times in the condition number of $\bm{W}_\lambda+ \frac{\Delta}{\gamma}\bm{I}$ up to logarithmic factor, i.e. $\tilde{O}(\gamma/\Delta) \times \polylog(1/\delta)$ times~\citep{ambainis:LIPIcs.STACS.2012.636,doi:10.1137/16M1087072,gilyen2019quantum}, where $\delta$ is chosen appropriately for the desired accuracy.
    In summary, Algorithm~\ref{alg: quantum algorithm} prepares a quantum state proportional to~\eqref{eq: phi def in quantum}, applies the quantum ridgelet transform $\bm{R}$ to the state, and then, applies $(\bm{W}_\lambda+ \frac{\Delta}{\gamma}\bm{I})^{-1/2}$, which yields a quantum state $\ket{p_{\lambda,\Delta}^\ast}$ in~\eqref{eq: optimized probability distoribution state}.
    This completes the proof of correctness of Algorithm~\ref{alg: quantum algorithm}.

    The runtime of Algorithm~\ref{alg: quantum algorithm} is dominated by applying $(\bm{W}_\lambda+ \frac{\Delta}{\gamma}\bm{I})^{-1/2}$.
    To implement this, we require $\tilde{O}(\gamma/\Delta)\times \polylog(1/\delta)$ preparations of $\bm{R}\ket{\hat{\phi}}$.
    Therefore, the total runtime is
    \begin{equation}
        O(D \; \polylog(MP)) \times \tilde{O}\left(\frac{\gamma}{\Delta}\right) \times \polylog\left(\frac{1}{\delta}\right) = O\left(D\;\polylog\left(\frac{MP}{\delta}\right)\right) \times \tilde{O}\left(\frac{\gamma}{\Delta}\right).
    \end{equation}
\end{proof}

\section{Details of Numerical Experiments}\label{sec: detail numerical}
We describe the detailed setting used in numerical experiments in Figures~\ref{fig: Comparison of empirical main} and~\ref{fig: Comparison of runtimes}.
All experiments are performed in the finite-domain setting over $\mathbb{Z}_P$.

\subsection{Empirical-risk experiment}
For the empirical-risk experiment, we set
\begin{align}
    P = 7, \quad D\in \{1,2,3,4,5,6\}, \quad M = 50 D.
\end{align}
For each dimension, we generate the training inputs independently from the uniform distribution over $\mathbb{Z}_P^D$, i.e.,
\begin{equation}
    p_\mathrm{data}(\bm{x}) = P^{-D}.
\end{equation}
The empirical distribution $\hat{p}_\mathrm{data}$ is constructed from these training data.
The labels are noiseless and are generated as $y_m = f_D(\bm{x}_m)$, where
\begin{equation}
    f_D(\bm{x}) = \sin\left( \frac{4\pi}{P} \left(\sum_{d = 1}^D x_d \bmod P \right) \right).
\end{equation}

We use a centered and normalized discrete ReLU activation function, i.e.,
\begin{equation}
    g(x) = r(x) \propto \begin{cases}
        x \bmod P, & 0 \leq (x \bmod P) \leq (P -1) /2, \\
        0. & (P + 1)/2 \leq  (x \bmod P) \leq P - 1.
    \end{cases}
\end{equation}
Figure~\ref{fig: sine function and normalized ReLU} shows the function $f_D$ and the activation function $g$ for the case $D = 1, P = 7$.
\begin{figure}[t]
  \centering
  \includegraphics[width=0.8 \linewidth]{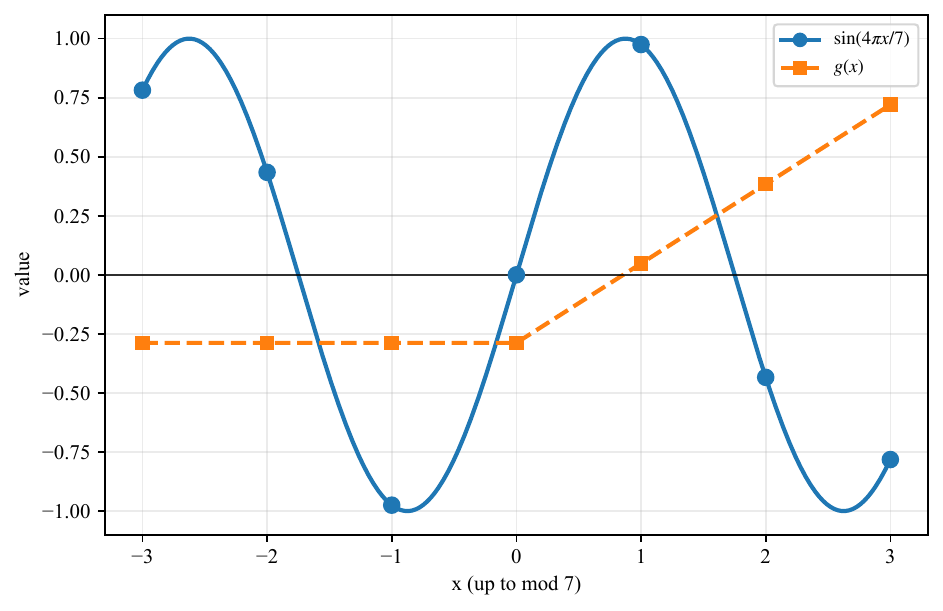}
  \caption{The function $f_D$ to be learned and the activation function $g$, when the dimension $D = 1$ and $P = 7$.}
  \label{fig: sine function and normalized ReLU}
\end{figure} 

The regularization parameter is set to 
\begin{equation}
    \lambda P^{-D} = 10^{-3},
\end{equation}
and we set
\begin{equation}
    \Delta = \gamma.
\end{equation}

We compare the following four sampling procedures:
\begin{itemize}
    \item exact sampling from the optimized distribution $p_{\lambda,\Delta}^{\ast}$,
    \item our dequantized sampler with accuracy parameter $\delta = 0.1$,
    \item our dequantized sampler with accuracy parameter $\delta = 0.7$,
    \item uniform sampling over $\mathbb{Z}_P^D \times \mathbb{Z}_P$.
\end{itemize}
For the exact optimized sampler, we explicitly enumerate all $(\bm{a},b) \in \mathbb{Z}_P^D \times \mathbb{Z}_P$ and compute the unnormalized weight
\begin{equation}
        s(\bm{a},b) = \left| \sum_{\bm{x}} P^{-D}r\left((\bm{a}^\top \bm{x} - b )\bmod P \right) \frac{\hat{p}_{\mathrm{data}}(\bm{x})f(\bm{x})}{\lambda P^{-D}+ \hat{p}_{\mathrm{data}}(\bm{x})} \right|^2,
\end{equation}
then compute the optimized distribution
\begin{equation}
    p_{\lambda,\Delta}^\ast(\bm{a},b)
    \coloneq
    \frac{1}{Z}
    \frac{s(\bm{a},b)}
         {s(\bm{a},b) + \Delta}.
\end{equation}
For each sampling procedure, we sample
\begin{equation}
    N \in \{ 8, 16, 32, 64, 128, 256, 512, 1024, 2048, 4096 \}
\end{equation}
nodes including overlaps.

For each sampling procedure, we first obtain a list of $N$ sampled hidden nodes
\begin{equation}
    T_N = \left( (\bm{a}_1,b_1),\ldots,(\bm{a}_N,b_N) \right) \in \left(\mathbb{Z}_P^D \times \mathbb{Z}_P\right)^N .
\end{equation}
Since the sampling procedures output nodes with replacement, the same hidden node
may appear multiple times in $T_N$.
Before fitting the output weights, we remove duplicated hidden nodes and define
\begin{equation}
    S_N \coloneq \operatorname{unique}(T_N) = \left\{ (\tilde{\bm{a}}_1,\tilde b_1),\ldots, (\tilde{\bm{a}}_{N_{\mathrm{eff}}},\tilde b_{N_{\mathrm{eff}}}) \right\},
\end{equation}
where
\begin{equation}
    N_{\mathrm{eff}} \coloneq |S_N| \leq \min\{N,P^{D+1}\}.
\end{equation}
Thus, $N$ is the number of sampled proposals, whereas $N_{\mathrm{eff}}$ is the
actual number of distinct hidden nodes used in the fitted sparse network.

Let
\begin{equation}
    \supp(\hat p_{\mathrm{data}}) = \{\bm{x}_1,\ldots,\bm{x}_K\}
\end{equation}
be the empirical support.
We write
\begin{equation}
    \hat p_i \coloneq \hat p_{\mathrm{data}}(\bm{x}_i), \qquad f_i \coloneq f_D(\bm{x}_i), \qquad \bm{f}_{\mathrm{supp}} \coloneq (f_1,\ldots,f_K)^\top,
\end{equation}
and
\begin{equation}
    \hat{\bm{P}}_{\mathrm{supp}}
    \coloneq
    \operatorname{diag}(\hat p_1,\ldots,\hat p_K).
\end{equation}
For the deduplicated sampled set $S_N$, we define the sampled design matrix
\begin{equation}
    \Phi_{S_N}
    \in
    \mathbb{R}^{K\times N_{\mathrm{eff}}}
\end{equation}
by
\begin{equation}
    (\Phi_{S_N})_{ij}
    \coloneq
    P^{-D/2}
    g\left(
        (\tilde{\bm{a}}_j^\top \bm{x}_i-\tilde b_j)\bmod P
    \right),
    \qquad
    i\in[K],\ j\in[N_{\mathrm{eff}}].
\end{equation}

The output weights are obtained by ridge regression on this sampled design
matrix:
\begin{equation}
    \hat{\theta}_{S_N}
    \in
    \operatorname*{arg\,min}_{\theta\in\mathbb{R}^{N_{\mathrm{eff}}}} 
    \left\|
        \hat{\bm{P}}_{\mathrm{supp}}^{1/2}
        \left(
            \Phi_{S_N}\theta-\bm{f}_{\mathrm{supp}}
        \right)
    \right\|_2^2
    +
    \lambda P^{-D}\|\theta\|_2^2 .
\end{equation}
The reported finite empirical risk is
\begin{equation}
    \widehat{\mathcal{R}}(S_N)
    \coloneq
    \left\|
        \hat{\bm{P}}_{\mathrm{supp}}^{1/2}
        \left(
            \Phi_{S_N}\hat{\theta}_{S_N}
            -
            \bm{f}_{\mathrm{supp}}
        \right)
    \right\|_2^2 =
    \sum_{i=1}^K
    \hat p_{\mathrm{data}}(\bm{x}_i)
    \left(
        (\Phi_{S_N}\hat{\theta}_{S_N})_i
        -
        f_D(\bm{x}_i)
    \right)^2 .
\end{equation}
We solve the corresponding normal equations using linear algebra routines in
NumPy.

Deduplication is important in low-dimensional experiments.
If identical hidden nodes were kept as separate columns of the design matrix, the same ridge function could be assigned multiple independent output weights, i.e., weights on the nodes may be non-unique.
Deduplicating the sampled nodes ensures that the fitted model has one output weight per distinct hidden node.

In Figure~\ref{fig: Comparison of empirical main}, the horizontal axis reports
the number of sampled proposals $N$.
The actual number of fitted hidden nodes is $N_{\mathrm{eff}}$, which is at most $P^{D+1}$.
Therefore, when $N>P^{D+1}$ in low-dimensional settings, the fitted network saturates at no more than $P^{D+1}$ distinct hidden nodes.

\subsection{Runtime experiment}
For the runtime experiment, we use the same data-generation procedure, target function, activation function, and regularization convention as in the empirical-risk experiment, except that we set $P = 3$.
This choice is made only to keep the naive baseline executable.
The hidden-node parameter space has size
\begin{equation}
    |\mathbb{Z}_P^D \times \mathbb{Z}_P| = P^{D + 1}.
\end{equation}
Therefore, explicitly constructing matrices over the hidden-node space becomes rapidly infeasible as $D$ grows.
For example, for $P = 7$ and $D = 7$, the hidden-node space already has size
\begin{equation}
    7^8 = 5,764,801,
\end{equation}
whereas for $P = 3$ and $D = 8$, it has size 
\begin{equation}
    3^9 = 19,683.
\end{equation}

The naive sampler explicitly constructs the matrix
\begin{equation}
    \bm{R}\hat{\bm{P}}_{\mathrm{data}}\bm{R}^\top + \lambda P^{-D} \bm{I},
\end{equation}
and solves the dense linear system
\begin{equation}
    \left(\bm{R}\hat{\bm{P}}_{\mathrm{data}}\bm{R}^\top + \lambda P^{-D} \bm{I}\right) w = \bm{R}\hat{\bm{P}}_\mathrm{data} \bm{f}.
\end{equation}
After computing $w$, it computes and samples the corresponding sampling weights over all $(\bm{a},b) \in \mathbb{Z}_P^D \times \mathbb{Z}_P$.

We run this naive sampler for
\begin{equation}
    D \in \{ 1,2,3,4,5,6,7,8\},
\end{equation}
or up to the largest dimension for which the dense linear solve is executable within memory and time limits.

We also run the classical sampler proposed in Algorithm~\ref{alg: ridgelet sampling}.
We measure the time required to output one sample.
This proposed sampler is run for
\begin{equation}
    D \in \{1,2,4,8,16,32,64,128\}.
\end{equation}

For each $D$ and each method, we repeat the runtime measurement over 20 independent executions.
Each point in Figure~\ref{fig: Comparison of runtimes} shows the mean, and each error bar shows the 95\% confidence interval of the mean.

All numerical experiments were conducted on a MacBook Pro with an Apple M4 Pro CPU and 48 GB unified memory.
The experiments were run on CPU using Python and NumPy, and the reported runtimes are wall-clock times measured by \texttt{time.perf\_counter()}.

\section{Limitations and Broader Impact}\label{sec: limitations-impact}
\subsection{Limitations}
Our theoretical results address the optimized ridgelet-based hidden-node sampling task for discretized single-hidden-layer neural networks.
They do not directly imply an efficient lottery-ticket discovery method for general deep neural networks.

The polynomial-time guarantees hold in the regime where the relevant parameters are polynomially bounded.
In particular, the runtime depends on the number of training samples $M$, the discretization parameter $P$ through polylogarithmic factors, the inverse smoothing parameter $\Delta^{-1}$, the weight-scale parameter $\gamma$, and the sampling accuracy through $\log(1/\delta)$.

The numerical experiments are intended as controlled demonstrations of the sampling procedure.
They use synthetic data for regression tasks.
Therefore, the experiments do not establish performance on real-world datasets or deep neural-network architectures.
Extending the method beyond the discretized ridgelet setting, in particular to deep neural-network architectures, remains an important direction for future work.

\subsection{Broader Impact}
This work is a foundational study of a sampling procedure for discretized single-hidden-layer neural-network models.
It does not introduce a deployed system, a pretrained model, a generative model, or a dataset involving human subjects.

We do not identify any direct societal impacts, positive or negative, from this work as presented.
Any societal implications would likely arise only from application-specific adaptations of the method, which should be evaluated in the context of those downstream applications.
\end{document}